\documentclass[11pt, a4paper, onecolumn, copyright, gdm]{google}

\usepackage{graphicx}

\usepackage{amsmath,amssymb,amsthm}
\allowdisplaybreaks
\usepackage{booktabs}

\usepackage{cleveref}
\usepackage{tabularx}
\usepackage{booktabs}
\usepackage{xcolor}
\usepackage{fancyvrb}
\usepackage{pgfplots}
\usepackage{multirow}
\pgfplotsset{compat=1.18}

\usepackage{listings}
\usepackage[scale=0.9]{inconsolata}
\usepackage[T1]{fontenc}
\usepackage{enumitem}

\usepackage{tikz}
\usetikzlibrary{positioning, arrows.meta, calc}

\newtheorem{theorem}{Theorem}

\newcommand{\xtarget}{x_t}

\definecolor{colorSDP}{HTML}{1D4ED8}
\definecolor{colorPraline}{HTML}{DC2626}
\definecolor{colorMC}{HTML}{16A34A}
\definecolor{colorDet}{HTML}{374151}

\usepackage[authoryear, sort&compress, round]{natbib}
\uselogo{} 

\title{Efficient and Sound Probabilistic Verification for AI Agents}

\correspondingauthor{alaia@engineering.upenn.edu \& dvij@google.com}
\affiliationnote{*Work done as a Student Researcher at Google DeepMind.}

\author[1,3,*]{Alaia Solko-Breslin}
\author[1]{Pramod Kaushik Mudrakarta}
\author[2]{Mihai Christodorescu}
\author[2,4]{Somesh Jha}
\author[1]{Krishnamurthy Dj Dvijotham}

\affil[1]{\thepa{}{}}
\affil[2]{Google}
\affil[3]{University of Pennsylvania}
\affil[4]{University of Wisconsin–Madison}

\begin{abstract}

Securing AI agents that operate in complex digital environments has become a critical need, and runtime monitoring approaches that formulate and enforce policies expressed in a formal language like Datalog offer a promising solution. However, existing approaches are restricted to deterministic policies. In many practical applications of AI agents, there is a need to enforce security policies in the face of ambiguity, leading to probabilistic predicates or state transitions (for example, a declassifier or Personally Identifiable Information (PII) detector that has some failure probability on each invocation). Furthermore, in many such applications, one cannot easily make the independence assumptions necessary to invoke prior work on probabilistic inference in Datalog. We address this by introducing a sound and efficient framework for such verification based on distributionally robust optimization, computing sound upper bounds on the probability of policy violation regardless of possible correlations between predicates.  On standard benchmarks for terminal and tool calling agents, we demonstrate that our approach outperforms prior art and improves the security-utility trade-off while ensuring rigorous bounds on the probability of policy violation.

\end{abstract}

\begin{document}

\maketitle

\section{Introduction}

The emergence of Large Language Models (LLMs) has enabled a new computing paradigm centered on AI agents that execute complex, multi-step tasks by acting directly within system environments \citep{yao2023react,schick2023toolformer,park2023generativagents}.
To fulfill user requests, these agents are given programmatic access to system interfaces, including terminals,  file systems, and external network APIs.
Despite its utility, this architecture introduces critical security and privacy risks by granting agents the ability to autonomously read, modify, and transmit data across system boundaries.

Consider how an agent that assists an enterprise employee with emailing contracts to a third-party vendor might accidentally share confidential data from the employee's filesystem that the vendor is not authorized to receive. Agents are susceptible to such security failures even in completely trusted environments due to execution or reasoning errors, but the threat compounds dramatically in untrusted settings, where adversarial tactics like prompt-injection attacks can override the agent's logic to fundamentally alter its planned trajectory~\citep{zhang2025agent,zhan2024injecagent,debenedetti2024agentdojo}.
Recent discoveries of successful real-world exploits validate these threats \citep{microsoft2026rce-vulnerabilities,bourtoule2026isolationoldvulnerabilities}, demonstrating that without guardrails, attackers can systematically hijack agent toolchains to silently exfiltrate sensitive user data.

These risks motivate the development of runtime defenses to enforce policy compliance during agent execution.
To provide formal enforcement guarantees, recent frameworks adapt the classical reference monitor paradigm to intercept proposed tool calls before execution~\citep{chen2025shieldagent,wang2026agentspec,palumbo2026formalpolicyenforcementrealworld}.
These monitors evaluate specific environmental predicates to make a decision to allow or block the action, e.g., only allowing a data-sharing action if it is determined that the target file does not contain sensitive data. However, in practical settings, an agent's environment is inherently ambiguous, for example from the imperfect accuracy of upstream classifiers and redaction tools. This persistent uncertainty necessitates a paradigm shift toward probabilistic verification, where policy compliance must be formally evaluated over probabilistic predicates.

A straightforward way to integrate probabilistic predicates into deterministic monitors would be to apply a threshold on the probability (say, the probability that a document contains PII).
However, this discretization introduces a fragile trade-off between security and utility.
For predicates representing the probability that an object is sensitive, high thresholds risk under-blocking harmful actions, and low thresholds risk over-blocking safe actions.
Ultimately, forcing continuous probabilities into binary choices prematurely discards the exact context needed to guide enforcement decisions, leaving deterministic monitors incapable of handling real-world ambiguity; probabilistic verification must therefore preserve this context.
To properly balance security and utility, the monitor should evaluate the current trajectory (i.e., its history of actions and data accesses) to compute a conservative upper bound on the global probability of the agent entering an unsafe state.
This computed upper bound is then compared against a global safety threshold that dictates whether the monitor blocks or allows the proposed action.

We formalize the computation of this global risk bound as probabilistic inference over an execution trajectory, where each tool invoked by the agent implements a specific state propagation contract.
These propagation rules can be either deterministic, such as copying a file to a new destination, or stochastic, such as noisy data declassification. 
The initial facts fed into these rules are parameterized by lower and upper marginal probability bounds, e.g., from a PII detector.
To track these dependencies across a multi-step trajectory, we express the state transitions as a Datalog program and compile the execution trace into a derivation graph supporting logical conjunction, disjunction, and inversion.

To evaluate execution risk, we model the Datalog derivation graph as an exact optimization problem over joint probability measures.
Specifically, we formulate an exponentially-sized linear program (LP) that enforces the marginal probabilities and state transition rules.
Importantly, this formulation makes no assumptions about the underlying joint distribution, remaining distributionally robust whether the input facts are statistically correlated or independent, unless such information is explicitly available as part of the model of the environment that the agent operates in.
This stands in direct contrast to standard probabilistic techniques like Monte Carlo sampling or Weighted Model Counting, which assume mutual independence and can severely underestimate risk during tool chaining. 
Because optimizing over the full state space of the exact LP is computationally intractable for real-time monitoring, we introduce a polynomially-sized semidefinite programming (SDP) relaxation.
Based on distributionally robust optimization (DRO) \citep{wiesemann2014dro}, this formulation tracks only second-order moments of the state distribution, yielding a sound over-approximation of the true worst-case probability of violation. This ensures that security is preserved, at the risk of a reduction in utility (which we empirically measure in our evaluation).

We evaluate our framework on terminal agent benchmarks, including Intercode-NL2Bash \citep{yang2023intercode} and ATBench \citep{liu2026agentdogdiagnosticguardrailframework}.
For these environments, we formalize taint transition semantics across tool calls to bound the worst-case probability of sensitive files being leaked during data-sharing actions.
To evaluate our SDP relaxation on other security policies, we also consider tasks from the Praline benchmarks \citep{wang2025praline}, including a side-channel vulnerability analysis.
Our results demonstrate that our approach outperforms prior art in terms of the security-utility tradeoff.

In summary, our key contributions are as follows:
\begin{enumerate}
    \item We identify security risks in deterministic verification engines in ambiguous agent environments and introduce a model for probabilistic verification. 
    \item We model multi-step agent trajectories via Datalog derivation graphs and formalize the computation of worst-case policy violation risk as an exact LP.
    \item We introduce a polynomially-sized SDP relaxation that tracks second-order moments to efficiently approximate risk at runtime, and we formally prove its soundness as a strict upper bound on execution risk.
    \item We evaluate our framework across terminal agent benchmarks, demonstrating that our relaxation effectively balances the security-utility tradeoff with low computational overhead.
\end{enumerate}

\section{Background and Motivation}

\subsection{Runtime Execution Monitoring for Agents}

Runtime execution monitoring uses a stateful reference monitor \citep{Anderson:1972} to enforce policies over an agent's execution trajectory by mediating tool access. The monitor acts as an inline interposition layer, intercepting candidate tool invocations before they execute or update the environment state.

\begin{figure*}
    \centering
    \includegraphics[width=\linewidth]{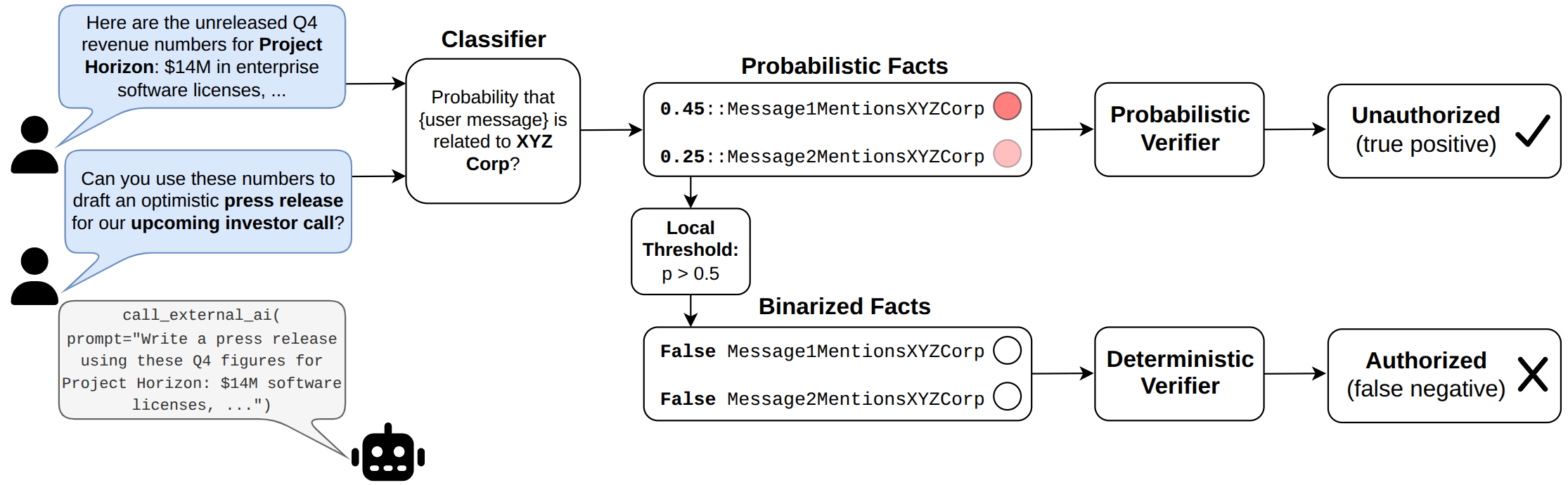}
    \caption{Overview of the deterministic and probabilistic verification paradigms for agents (based on the policy from Listing \ref{lsting:forge-policy}). Binary evaluation (e.g., FORGE \citep{palumbo2026formalpolicyenforcementrealworld}) requires applying local thresholds to probabilistic facts, leading to a loss of information. Despite both messages indirectly referring to XYZ Corp, their predicates are classified as False, resulting in the decision to authorize the tool (false negative). Conversely, probabilistic verification determines that the cumulative probability of the messages mentioning XYZ Corp is high enough to warrant denial (true positive).}
    \label{fig:probabilistic-agent-verification}
\end{figure*}

Formally, we model the policy-relevant environmental context as a state $s \in \mathcal{S}$, initialized to a base configuration $s_0$. An execution trajectory is a sequence of tool calls $t_1, t_2, \dots, t_T$, where each $t_k \in \mathcal{T}$ specifies a tool identifier and its runtime arguments. The environmental effects of these interactions are governed by a state transition function:
$$\Delta: \mathcal{S} \times \mathcal{T} \rightarrow \mathcal{S}$$
which maps the current system state and a candidate tool call to a successor state.
Execution boundaries are specified via a policy predicate $\Pi: \mathcal{S} \rightarrow \{\text{Allow}, \text{Block}\}$, which determines whether a given state configuration is compliant.

The reference monitor enforces safety via an intercept-evaluate-commit cycle at each step $k$. When the agent issues a candidate tool call $t_k$ from state $s_k$, the monitor suspends execution and speculatively projects the next system state: $s'_{k+1} = \Delta(s_k, t_k)$. If $\Pi(s'_{k+1}) = \text{Block}$, the monitor blocks the tool invocation, discards the speculative update ($s_{k+1} \leftarrow s_k$), and returns structured feedback to the agent to trigger autonomous error recovery. Otherwise, the tool is executed and the state transition is committed ($s_{k+1} \leftarrow s'_{k+1}$).

\subsection{Probabilistic Execution Monitoring}\label{sec:probabilistic-execution-monitoring}

While deterministic reference monitors provide strong structural guarantees for discrete inputs, their enforcement capabilities degrade when policies rely on noisy or ambiguous predicates (say, the output of a learned classifier or LLM-based judge). 
Traditional monitors can handle this ambiguity by applying local thresholds to convert probabilistic outputs into binary values. 
However, this early binarization creates an information bottleneck, discarding marginal uncertainties that compound over multi-step trajectories into global policy violations.

To illustrate this vulnerability, consider a security policy from FORGE~\citep{palumbo2026formalpolicyenforcementrealworld} designed to prevent unauthorized data exfiltration. The policy blocks an agent's call to an external API if its contents depend on a user message that contains sensitive proprietary information:

\begin{lstlisting}[caption={Privacy policy for protecting sensitive proprietary information.}, label={lsting:forge-policy},language=Prolog]
Unauthorized(a) :-
    Actions(a),
    is_tool_call(a),
    Current(id),
    Depends(id, id2),
    SentMessage(id2, msg),
    msg.agent_role == User,
    llm_check("Message references XYZ Corp", msg.contents).
\end{lstlisting}
Because the variable \texttt{id2} can bind to any past message within the context window, this single rule can be expanded into a logical disjunction evaluating risk across the entire multi-turn conversation history. 

When a reference monitor relies on local thresholding, it forces a hard decision on each past message branch independently before evaluating the global policy.
As illustrated in \autoref{fig:probabilistic-agent-verification}, we consider an execution trace containing two consecutive user messages: $m_1$ (``Here are the unreleased Q4 revenue numbers for Project Horizon: \$14M in enterprise software licenses, ...'') and $m_2$ (``Can you use these numbers to draft an optimistic press release for our upcoming investor call?'').
Due to the ambiguity of internal project codenames and what is meant by the ``upcoming investor call,'' a classifier might return a marginal risk probability of $0.45$ for $m_1$ and $0.25$ for $m_2$.
If the agent decides to call an external API to summarize the Q4 figures for the press release, this prompt definitively depends on $m_1$ and $m_2$, so the security question becomes: does either message reference XYZ Corp?
If the reference monitor applies a local threshold of $0.50$ to each predicate, a deterministic engine would evaluate both predicates to \texttt{False}, underestimating the collective threat and allowing the unsafe execution.

Conversely, probabilistic verification preserves continuous weights and computes the joint probability of the entire query over the execution trace.
Therefore, if the reference monitor computes the cumulative probability of these predicates and applies a global threshold (e.g., $0.50$) to the final risk probability, it would correctly evaluate the policy as \texttt{Unauthorized} and block the unsafe execution.

Computing exact joint distributions over complex trajectories is often computationally intractable at runtime, but efficient over-approximation of the risk probability still preserves the system's structural safety guarantees. 
Let $\mathcal{R}(a) \in [0, 1]$ represent the exact risk probability that an action $a$ violates the policy based on the underlying predicates, and let $\tau$ be the global safety threshold. 
If the reference monitor applies a sound upper-bound estimate $\mathcal{R}'(a) \ge \mathcal{R}(a)$, the engine maintains a predictable operational bound:
$$\mathcal{R}'(a) \le \tau \implies \mathcal{R}(a) \le \tau$$
Consequently, the approximation introduces no structural false negatives relative to an exact evaluation. While noisy base classifiers can still propagate errors, the monitor itself adds no additional vulnerability.
Choosing an estimator thus requires balancing this safety guarantee with system utility, keeping the upper bound tight enough to avoid over-blocking that disrupts valid workflows.

\subsection{Threat Model}
We use the standard threat model of a security reference monitor~\citep{Anderson:1972}, in which the attacker has direct or indirect access to an AI agent and can influence the instructions and data that the agent receives. Note that this model allows for indirect prompt-injection attacks as well for the user to be malicious, both of which can result in the agent taking undesirable actions. The actions that the agent takes are in the form of tool calls, where each tool call is verified by a reference monitor for compliance with a security policy before the tool executes. The attacker cannot bypass or disable the reference monitor, nor can it modify the security policy, but has knowledge of the available tools, the security policy, and the tool failure rates. We however need to assume that the attacker cannot invalidate the probabilistic assumptions we make on tool failures
(see the ``Tool threat model'' discussion in \autoref{sec:derivation-state-probabilities}).

\subsection{Related Work}

\textbf{Agent guardrails and defenses.} Early defenses for agents rely on targeted prompting and alignment tuning to encourage policy compliance and prompt-injection resistance~\citep{schick2023toolformer, zheng2024promptdrivensafeguarding, ghalebikesabi2024operationalizingcontextualintegrityprivacyconscious, chen2025struQ, chen2025secalign}. To enforce stronger boundaries, structural defenses isolate untrusted inputs from the agent's planning layer entirely~\citep{debenedetti2025defeatingpromptinjectionsdesign, li2026drift}, while external guardrail toolkits use standalone classification models to filter queries and tool responses~\citep{inan2023llamaguardllmbasedinputoutput, rebedea-etal-2023-nemo}. To avoid the latency overhead of these always-on filters, CausalArmor~\citep{kim2026causalarmorefficientindirectprompt} uses ablation-based attribution to detect when an untrusted input overrides the user's original intent, triggering targeted sanitization and retroactive thought-masking only when a threat is present. While these strategies have been empirically shown to mitigate common attack vectors, they operate heuristically and lack the formal enforcement mechanisms needed to guarantee policy compliance across agent trajectories.

\textbf{Formal policy enforcement.} 
Several frameworks adapt classic reference monitors~\citep{Anderson:1972, crampton2005referencemonitor} to deterministically enforce safety policies over an agent's execution trajectory, aligning with foundational models of execution monitoring~\citep{schneider2000enforceablesecuritypolicies}.
These approaches compile formal specifications, such as temporal logic, regular languages, or behavioral contracts, into runtime guards that intercept untrusted actions before execution~\citep{li2024formalllmintegratingformallanguage, chen2025shieldagent, kamath2025enforcingtemporalconstraintsllm, miculicich2025veriguardenhancingllmagent, tsai2025conesca, palumbo2026formalpolicyenforcementrealworld, wang2026agentspec}.
A parallel line of work focuses on dynamic authorization, enforcing the principle of least privilege by restricting tool-calling capabilities as the execution context evolves~\citep{zhu2025miniscopeprivilegeframeworkauthorizing, shi2026progentsecuringaiagents}. 
While these deterministic techniques provide rigorous guardrails, they are incompatible with noisy, probabilistic predicates unless ad-hoc thresholds are applied.
Our framework addresses this limitation by extending provable enforcement to uncertain environments via distributionally robust risk estimation.

\textbf{Information flow control for agents.}
Privacy policies for autonomous agents frequently restrict what untrusted content can be read, processed, or shared with external recipients. 
This connects to classic Information Flow Control (IFC) literature, which uses security lattices and non-interference properties to ensure that secret data cannot influence public outputs~\citep{denning1976latticeifc, myers1997ifc, sabelfeld2003language}. 
Recent frameworks apply these traditional IFC properties directly to agent runtimes to prevent data exfiltration. 
For instance, system-level defenses disaggregate the architecture into context-aware pipelines to filter untrusted planning inputs~\citep{wu2024systemleveldefenseindirectprompt}, while language model planners leverage dynamic taint-tracking to enforce confidentiality and integrity labels across agent states~\citep{costa2025securingaiagentsinformationflow}.
While these techniques provide formal guarantees for agentic information flow control, they still rely on discrete security labels.
Our approach instead tracks the propagation of continuous sensitivity probabilities, drawing conceptual parallels to Quantitative Information Flow (QIF) frameworks~\citep{Denning1982-qb,gray1991informationflowsecurity}.
The  goal of QIF is to measure the amount of information leaked about a secret by observing the running of a program.
In this paper, the uncertainty arises from the use of non-deterministic components (e.g., using an LLM-based classifier), so the sources of uncertainty in our model are different from QIF.

\textbf{Probabilistic logic programming.}
Probabilistic logic programming frameworks for Prolog~\citep{deRaedt2007ProbLog,fierens2015inference,manhaeve2018deepproblog} and Datalog~\citep{barany2017datalogdppl,Tsamoura2020groundingbottleneckdatalog,li2023scallop}, alongside frameworks for weighted model counting over probabilistic programs~\citep{sang2005wmc-complexity,holtzen2020dice}, reason under uncertainty by combining formal logic with probabilistic predicates.
However, these frameworks assume that all base facts are independent, an assumption that fails to hold in agentic settings where environment states and tool outputs are inherently coupled.
While linear programming frameworks can exactly optimize over measures to handle these arbitrary dependencies~\citep{andersen1996linear}, the required optimization scales exponentially with the number of states modeled.
Even when specialized relational frameworks like JudgeD~\citep{wanders2016judgeD} and Praline~\citep{wang2025praline} compute this optimum within a small error bound, the underlying solver overhead is too inefficient for real-time verification over complex execution traces.

\section{Probabilistic Verification as an Optimization Problem}

To rigorously bound the probability that an agent's trajectory violates a given policy, we formulate verification as a distributionally robust optimization problem.
This framework requires two core components: a set of sound marginal probability bounds on the initial environmental facts, and a formal transition logic that propagates these facts across an execution trace.
In this section, we first detail how a security engineer derives these initial base bounds from empirical tools and calibrated classifiers.
We then formalize these inputs within a robust Probabilistic Datalog framework and construct the exact optimization problem over the resulting execution state measures.

\subsection{Probabilistic Models of Tools}\label{sec:derivation-state-probabilities}
We describe in this section how one might develop probabilistic models of certain tools in the environment that would then be used to compute the final certificate.

\textbf{Deterministic tool failure rates.} To bound the failure probability of a tool (say, an LLM-based or deterministic redaction tool, or an LLM-based PII detector), the security engineer profiles the tool on a representative evaluation dataset with known ground-truth labels.
Specifically, if the tool exhibits $k$ failures across $N$ independent evaluation trials, the empirical failure rate is $\hat{p} = \frac{k}{N}$.
The engineer can establish a sound upper bound $u_{\text{redact\_fail}}$ using a $1-\alpha$ Clopper-Pearson confidence interval:
\[
u_{\text{redact\_fail}} = F_{\beta}^{-1}\left(1-\alpha; k+1, N-k\right)
\]
where $F_{\beta}^{-1}$ represents the inverse cumulative distribution function of a Beta distribution.

\textbf{Correlations across multiple invocations.}
The above assumptions are made marginally, per tool invocation, but do not make any explicit assumptions about correlations between tool or classifier failure rates across multiple invocations in a single agent run (e.g., multiple calls to a PII detector tool). However, in the presence of further information, we can assume for example that the tool errors are positively correlated across multiple runs, or that failure rates for two complementary tools are negatively correlated. Following prior work~\citep{wang2025praline}, we model these as three possible correlation classes: $\text{POS}$ (positively correlated), $\text{NEG}$ (negatively correlated) and $\text{IND}$ (statistically independent).

\textbf{Tool threat model.} We assume that the above probabilistic models of tools are robust in the face of attackers, in other words, tool failures occur independently of attacker actions. This can be justified in one of two ways: 
1. \emph{Tools with probabilistic models execute in sanitized environments:} A tool failure occurs because of an exogenous event or that the inputs to calibrated classifiers only come from trusted sources that an attacker cannot influence.
2. \emph{Tools have intrinsic robustness:} For example, a text classifier that has been trained to be robust to specific perturbations that an attacker can make.

Further validation of this threat model remains out of scope for this paper; we take this assumption as a given in the rest of this paper.

\subsection{Probabilistic Inference in Datalog}\label{sec:probabilistic-datalog}

Probabilistic inference in Datalog begins with a program $\mathcal{D} = (\mathcal{F}, \mathcal{R})$, where $\mathcal{F}$ is a finite set of ground base facts and $\mathcal{R}$ is a set of ground intensional rules used to derive additional facts.
Unlike classical settings, we do not assume that the base facts $f \in \mathcal{F}$ represent independent random variables.
Instead, the true joint distribution over the facts is governed by an unknown probability measure $\mu$.
This acknowledges that in agent environments, base facts can be positively correlated, but their exact correlation structures are unknown to the verifier.
For example, a PII classifier might fail to detect the same piece of sensitive data common to multiple files; similarly, a redaction tool might fail to redact the same sensitive data across these files.

We assume that each base fact $f \in \mathcal{F}$ is associated with known or estimated marginal probability bounds $[\ell_f, u_f] \subseteq [0, 1]$.
For any specific realization of base facts $L \subseteq \mathcal{F}$, also referred to as a ``world,'' the complete deterministic program is evaluated as $L \cup \mathcal{R}$. 
Because ground Datalog programs guarantee a unique least fixed point, the semantics of the program are uniquely defined by its minimal model, denoted $\mathcal{M}(L \cup \mathcal{R})$. 
Let $\mathcal{B}_{\mathcal{D}}$ denote the Herbrand base of the program, encompassing all ground base and derived facts. The goal of distributionally robust probabilistic inference is to bound the success probability $P(q)$ of a safety query $q \in \mathcal{B}_{\mathcal{D}}$ over the set of all valid joint measures $\mu$ that respect these marginal constraints:
$$P(q) = \sum_{L \subseteq \mathcal{F}} \mu(L) \mathbb{I}[q \in \mathcal{M}(L \cup \mathcal{R})]$$

Because Datalog rules are monotone, the query $q$ can be represented as a symbolic lineage formula over the Boolean random variables $x_f$ associated with each base fact $f \in \mathcal{F}$. 
This formula forms a monotone Disjunctive Normal Form (DNF) over the base facts that derive $q$:
$$q \equiv \bigvee_{d \in \text{Derv}(q)} \bigwedge_{f \in \text{facts}(d)} x_f$$
where $\text{Derv}(q)$ represents the operational derivations of $q$ under $\mathcal{R}$, and $\text{facts}(d) \subseteq \mathcal{F}$.

\textbf{The risk of independence semantics.}
In contrast to our distribution-free model, most existing frameworks for probabilistic inference in Datalog closely follow the ProbLog model \citep{deRaedt2007ProbLog}, assuming that the base facts are mutually independent. Formally, for any base fact $f \in \mathcal{F}$ and any subset of remaining facts $S \subseteq \mathcal{F} \setminus \{f\}$, the conditional probability satisfies:$$P(x_f \mid \mathbf{x}_S) = P(x_f)$$where $\mathbf{x}_S$ denotes the joint assignment of the variables in $S$.
Under this assumption, the probability of sampling any single subprogram $L$ collapses to a simple product of marginal probabilities:
$$P(L \mid \mathcal{D}) = \prod_{f \in L} p_f \prod_{f \in \mathcal{F} \setminus L} (1 - p_f)$$
where $P(x_f) = p_f$. Evaluating $P(q \mid \mathcal{D})$ exactly under this model is equivalent to Weighted Model Counting (WMC), which is known to be \#P-hard \citep{sang2005wmc-complexity}. 

More importantly, the independence assumption is fundamentally unsafe for policy verification. For example, if an agent applies the same redaction tool (with failure probability $p$) twice to the same file, WMC assumes independence and computes the joint failure probability as $P(x_{f_1} \land x_{f_2}) = p^2$.
In reality, these steps are perfectly correlated: if the tool misses a piece of PII on the first run, it will also fail on the second, so the true risk remains $p$.
Distributionally robust verification avoids this risk by making no assumptions about these joint correlations.

\subsection{Compiling Datalog Trajectories to a DAG}\label{sec:datalog-compilation}

To evaluate policy compliance over an agent's runtime behavior, we translate its execution trace into a concrete Datalog program $\mathcal{D}$, where $\mathcal{R}$ operationalizes the policy's state transition logic and $\mathcal{F}$ captures the ground observations.
This translation manifests as a directed acyclic (derivation) graph (DAG) $G = (V_F \cup V_R, E)$, where the rule nodes $V_R$ represent the instantiated rules $\mathcal{R}$, and the fact nodes $V_F$ encompass both the initial base facts $\mathcal{F}$ and all subsequent derived facts.
The edges $E$ connect parent facts to child facts via their corresponding rule nodes. This graph structure allows us to decompose the global lineage of a given query into a localized system of algebraic constraints.

Consider an agent tasked by an employee at an enterprise with declassifying \texttt{a.txt}, combining it with \texttt{b.txt}, copying the result to \texttt{res.txt}, and sending the outcome to \texttt{alex@external.com}, a third-party vendor. A set of commands the agent might execute to accomplish this task is as follows:
\begin{lstlisting}
    redact a.txt > i1.txt
    cat b.txt i1.txt > i2.txt
    cp i2.txt res.txt
    send_file res.txt alex@external.com
\end{lstlisting}
The security objective is to determine whether sensitive information could get leaked to the third-party vendor. Given the above trace, this reduces to checking whether \texttt{res.txt} contains sensitive data.

The compilation begins by establishing the base facts $f \in \mathcal{F}$, which represent the source nodes in the graph (having no parent facts).
We initialize the graph with prior marginal probabilities for these base facts.
For example, we might run a PII classifier on \texttt{a.txt} and \texttt{b.txt}, producing probabilities for the facts \texttt{sensitive(a)} and \texttt{sensitive(b)}.
We also start with the probability that \texttt{redact} fails to remove sensitive data, taken from historical failure rates of the redaction tool.
Following standard convention, we prefix each fact $f$ with its marginal probability $p_f$ using the notation $p_f::f$:
\begin{lstlisting}[language=Prolog]
    0.9::sensitive(a).
    0.2::sensitive(b).
    0.1::redact_fail.
\end{lstlisting}
These atomic declarations instantiate the initial set of fact nodes $V_F$.
As the agent executes commands sequentially, each step appends individual rules $r$ to the set of intensional rules $\mathcal{R}$, thereby deriving fact nodes whose truth values depend on upstream predicates.
Note that these rules are predefined by a formal policy which assigns information-flow semantics to each available tool (\autoref{tab:taint-semantics}).
We now describe how each operation in this trajectory maps to transition rules which guide the construction of the DAG.

\begin{table}[t]
\centering
\renewcommand{\arraystretch}{1.1}
\normalsize
\caption{Taint semantics for each transition type. Each bash command is assigned one of these transition types.}
\begin{tabular}{lll}
\toprule
\textbf{Transition} & \textbf{Semantics} \\
\midrule
\textsc{Propagate} & $x_v = x_u$ \\
\textsc{Merge} & $x_v = x_{u_1} \vee x_{u_2}$ \\
\textsc{Declassify} & $x_v = x_u \land x_{\text{redact\_fail}} $ \\
\textsc{CreateClean} & $x_v = \text{false}$ \\
\textsc{CreateTainted} & $x_v = \text{true}$ \\
\bottomrule
\end{tabular}
\label{tab:taint-semantics}
\end{table}

\textbf{Data declassification.}
The agent's first command attempts to strip sensitive information from \texttt{a.txt}.
The semantics of this transformation are captured by a \textsc{Declassify} transition rule:
\begin{lstlisting}[language=Prolog]
    sensitive(i1) :- sensitive(a), 
                     redact_fail.
\end{lstlisting}
This rule establishes a conjunctive rule node $r_1 \in V_R$.
The derived fact node $v_{i1} \in V_F$ becomes a child of $r_1$, which in turn draws its incoming edges from its structural predecessors $\text{pre}(r_1) = \{v_a, v_{\text{redact\_fail}}\}$.
The intermediate file \texttt{i1.txt} is deemed sensitive if the source file was originally sensitive and the redaction failed to remove all sensitive data.

\textbf{Data merging.}
Next, the agent combines the contents of \texttt{b.txt} and the redacted intermediate file, representing a \textsc{Merge} transition:
\begin{lstlisting}[language=Prolog]
    sensitive(i2) :- sensitive(i1).
    sensitive(i2) :- sensitive(b).
\end{lstlisting}
In contrast to the previous step, this operation introduces a logical disjunction.
It instantiates two distinct rule nodes $r_2, r_3 \in V_R$ that both point to the same derived fact node $v_{i2} \in V_F$. 
Consequently, \texttt{i2.txt} inherits sensitivity if either upstream file is sensitive.

\textbf{Data propagation.}
In the last step before \texttt{send\_file}, the agent copies the merged contents to the final outbound file target, representing a \textsc{Propagate} transition:
\begin{lstlisting}[language=Prolog]
    sensitive(res) :- sensitive(i2).
\end{lstlisting}
This rule simply asserts that $v_{res}$ is sensitive if its immediate predecessor $v_{i2}$ is sensitive.

\textbf{Querying the outbound file.}
The trajectory terminates with an external data transmission.
Because this action crosses a trust boundary, the security policy marks the target attribute as a terminal query node $q \in V_F$:
\begin{lstlisting}[language=Prolog]
    query(sensitive(res)).
\end{lstlisting}
This step-by-step compilation yields a complete derivation graph $G$ (\autoref{fig:derivation_graph_example}) which preserves the exact sensitivity provenance chain resulting from the agent's actions.

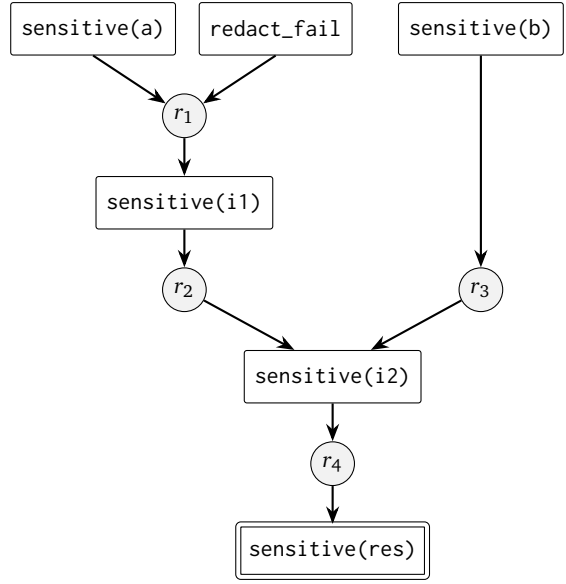
\begin{figure}[t]
\centering
\begin{tikzpicture}[
    fact/.style={
        draw, 
        rectangle, 
        rounded corners=1pt,
        font=\small\ttfamily, 
        inner sep=4pt,
        minimum height=1.8em
    },
    rule/.style={
        draw, 
        circle, 
        fill=gray!10, 
        font=\footnotesize, 
        inner sep=1pt,
        minimum size=1.5em
    },
    query/.style={
        draw, 
        double, 
        double distance=1.5pt,
        rectangle, 
        rounded corners=1pt,
        font=\small\ttfamily, 
        inner sep=4pt,
        minimum height=1.8em
    },
    edge/.style={
        ->, 
        >=Stealth, 
        thick,
        draw=black
    }
]
    \node[fact] (sa) {sensitive(a)};
    \node[fact] (rf) [right=0.3cm of sa] {redact\_fail};
    \node[fact] (sb) [right=0.6cm of rf] {sensitive(b)};

    \path (sa.south) -- (rf.south) coordinate[midway] (r1_input);
    \node[rule] (r1) [below=0.5cm of r1_input] {$r_1$};

    \node[fact] (si1) [below=0.5cm of r1] {sensitive(i1)};

    \node[rule] (r2) [below=0.5cm of si1] {$r_2$};
    \node[rule] (r3) at (sb |- r2) {$r_3$};

    \path (r2.south) -- (r3.south) coordinate[midway] (r23_output);
    \node[fact] (si2) [below=0.5cm of r23_output] {sensitive(i2)};

    \node[rule] (r4) [below=0.5cm of si2] {$r_4$};
    \node[query] (sres) [below=0.5cm of r4] {sensitive(res)};

    \draw[edge] (sa.south) -- (r1);
    \draw[edge] (rf.south) -- (r1);
    \draw[edge] (r1) -- (si1);

    \draw[edge] (si1.south) -- (r2);
    \draw[edge] (r2) -- (si2.145);
    
    \draw[edge] (sb.south) -- (r3);
    \draw[edge] (r3) -- (si2.35);

    \draw[edge] (si2.south) -- (r4);
    \draw[edge] (r4) -- (sres);

\end{tikzpicture}
\caption{Compiled derivation graph for deriving the sensitivity of \texttt{res}. Fact nodes are boxed, and the query node is double-boxed. Rule nodes used to derive new fact nodes are circled.}
\label{fig:derivation_graph_example}
\end{figure}

Through this derivation graph, we embed the Boolean logic into the real space $\mathbb{R}^{|V_F|}$, mapping each logical gate to its equivalent multilinear polynomial constraint over $\{0, 1\}^{|V_F|}$. Conjunctions map to multiplication $(x_1x_2)$, disjunctions $(x_1 \lor x_2)$ map to inclusion-exclusion $(x_1 + x_2 - x_1x_2)$, and negations $(\lnot x_1)$ map to inversion $(1 - x_1)$.

Applying this mapping to our example yields the following Boolean constraints for each derived fact:
\begin{align*}
    x_{i1} &= x_a x_{\text{redact\_fail}} \\
    x_{i2} &= x_{i1} + x_b - x_{i1}x_b \\
    x_{\text{res}} &= x_{i2}
\end{align*}
Rather than collapsing this chain into a single algebraic derivation, we preserve these intermediate variables to ensure compatibility with a lower-order relaxation.
Formally, for any derived fact $v \in V_F$, its logical parent set $\mathcal{P}(v) \subseteq V_F$ comprises the fact nodes that feed into $v$'s immediate rules.
Consequently, the Boolean transition for any derived fact $v \in V_F$ is expressed as a multilinear polynomial over its immediate parent facts:
$$x_v = \phi_v(x_{\mathcal{P} (v)}).$$

\begin{theorem}\label{thm:Arity2}
For any security policy written in Datalog, its execution on any finite trace of events can be rewritten as a derivation graph where $|\mathcal{P}(v)|\leq 2$ for all $v \in V_F$.
\end{theorem}
\begin{proof}
    This can be shown by induction on the number of immediate parent nodes $n = |\mathcal{P}(v)|$ of a derived fact $v$. In the base case ($n \le 2$) the condition holds trivially.
    In the inductive case, we assume that any fact with $n$ immediate parents can be rewritten into a derivation graph with maximum arity 2.
    In the case of multi-rule disjunction, if the node $v$ is derived by $n+1$ distinct rules, we can isolate the first $n$ rules into an auxiliary node $x_{\text{aux}}$, yielding the binary relation $x_v = x_{\text{aux}} \lor x_{u_{n+1}}$.
    In the case of multi-literal conjunction, if the node $v$ is derived by a Horn clause rule with $n+1$ body literals, we can isolate the first $n$ literals into an auxiliary node $x_{\text{aux}}$, yielding the binary relation $x_v = x_{\text{aux}} \land x_{u_{n+1}}$.
    Both cases follow from the associativity of $\lor$ and $\land$.
\end{proof}

\subsection{Exact Optimization over State Measures}\label{sec:exact-optimization-measures}

After compiling the trajectory, we are left with a DAG where each node represents a state variable.
Let the collection of these variables be denoted $x = (x_1, \dots, x_{|V_F|}) \in \{0,1\}^{|V_F|}$.
Recall from \autoref{sec:probabilistic-datalog} that $\mu$ denotes the joint probability measure over the base facts. 
Because all derived variables are deterministic functions of these base facts, $\mu$ uniquely extends to a joint distribution over the full state vector $x \in \{0,1\}^{|V_F|}$
Additionally, there is a set of unsafe states $\mathcal{U}\subseteq \{0, 1\}^{|V_F|}$, e.g., states in which a set of outbound files are considered sensitive. In this paper, we will assume that $\mathcal{U}$ has the form $\{x: \xtarget=1\}$

The probability of an unsafe execution under a given joint distribution $\mu$ is the measure of this unsafe set:
\[R(\mu) = \sum_{x \in \mathcal{U}} \mu(x)=\sum_{x: \xtarget=1} \mu(x).\]
To bound this risk rigorously, we define $R^*$ as the maximum probability of entering an unsafe state.
We compute this bound by seeking a joint probability measure $\mu$ on $\{0,1\}^{|V_F|}$ that is consistent with all state transitions implied by the execution graph and initial conditions, but which maximizes the risk of unsafe execution.
Known initial conditions on atomic variables $x_v$ (where $\mathcal{P}(v) = \emptyset$) are bounded above and below by marginal probabilities $u_v$ and $\ell_v$, respectively.
We formulate an exact optimization problem over measures $\mu$ as follows:
\begin{equation}
\label{eq:exact-optimization}
\begin{split}
 R^* = \max_{\mu} \quad & \sum_{x \in \mathcal{U}} \mu(x) \\
 \text{s.t.} \quad & \sum_{x} \mu(x) \mathbb{I}[x_v = \phi_v(x_{\mathcal{P}(v)})] = 1 \\ 
 & \qquad\qquad \forall v \in V_F : \mathcal{P}(v) \neq \emptyset \\
&  \ell_v \leq \sum_{x} \mu(x) \mathbb{I}[x_v=1] \leq u_v \\
& \qquad\qquad  \forall v \in V_F : \mathcal{P}(v) = \emptyset \\
&  \mu(x) \geq 0 \quad \forall x \in \{0, 1\}^{|V_F|} \\
& \sum_{x \in \{0, 1\}^{|V_F|}} \mu(x) = 1
\end{split}
\end{equation}
The first constraint restricts the support of the measure $\mu$ to states that strictly satisfy the deterministic transition logic, forcing the probability of impossible states to zero.
The second constraint enforces the lower and upper marginal bounds on the initial base facts based on our prior observations. Finally, the remaining constraints guarantee that $\mu$ constitutes a valid probability distribution by ensuring non-negativity and total mass normalization.

\subsection{Exact Optimization Example}

We illustrate the construction of the exact optimization problem using our running example from \autoref{sec:datalog-compilation}.
The state vector is defined over our six fact variables: $x = (x_a, x_b, x_{\text{redact\_fail}}, x_{i1}, x_{i2}, \xtarget)$ (note that we relabel $x_\mathrm{res}$ to $\xtarget$ to have a consistent notation for the unsafe set). The unsafe set $\mathcal{U}$ corresponds to any state configuration where the outbound file is sensitive, meaning $\mathcal{U} = \{x \mid \xtarget = 1\}$. For the sake of simplicity, we assume that our marginal prior probabilities are tight point probabilities in this example, setting $\ell_v = u_v = p_v$ for all base facts $v \in V_F $ (where $\mathcal{P}(v) = \emptyset$).
Substituting these marginals and our local multilinear polynomials into \autoref{eq:exact-optimization} yields the following program:
\begin{align*}
  R^* = \max_{\mu} \quad & \sum_{x : \xtarget=1} \mu(x) \\
 \text{s.t.} \quad & \mu(x) = 0 \quad \forall x \text{ where } x_{i1} \neq x_a x_{\text{redact\_fail}} \\
 & \mu(x) = 0 \quad \forall x \text{ where } x_{i2} \neq x_{i1} + x_b - x_{i1}x_b \\
 & \mu(x) = 0 \quad \forall x \text{ where } \xtarget \neq x_{i2} \\
 & \sum_{x} \mu(x) x_a = 0.9, \\
 & \sum_{x} \mu(x) x_b = 0.2, \\
 & \sum_{x} \mu(x) x_{\text{redact\_fail}} = 0.1 \\
 & \sum_{x} \mu(x) = 1, \quad \mu(x) \geq 0 \quad \forall x.
\end{align*}

Solving this program yields a maximum leakage probability of $R^* = 0.300$ ($30.0\%)$. This upper bound can be understood analytically via Fréchet bounds. For the rule $x_{i1} = x_a \land x_{\text{redact\_fail}}$, the maximum probability of the intersection is $$\min(P(x_a), P(x_{\text{redact\_fail}})) = \min(0.9, 0.1) = 0.1,$$ achieved when the two facts are perfectly positively correlated, e.g., the redaction tool always fails to remove some type of PII contained in \texttt{a.txt}.
For the disjunction $x_{i2} = x_{i1} \lor x_b$, the optimizer maximizes the total probability by making $x_{i1}$ and $x_b$ mutually exclusive, which eliminates the overlapping penalty term and yields $0.1 + 0.2 = 0.3$.

In contrast, standard probabilistic inference for Datalog evaluates this query using Weighted Model Counting (WMC) as outlined in \autoref{sec:probabilistic-datalog}. WMC operates under the rigid assumption that all base facts are strictly independent, fixing the joint distribution to a product of marginal probabilities.
Evaluating the global lineage formula, $\xtarget \equiv (x_a \land x_{\text{redact\_fail}}) \lor x_b$, under independence corresponds to the following step-by-step calculation:
\begin{align*}
P(x_{i1} = 1) &= 0.9 \times 0.1 = 0.09 \\
P(\xtarget = 1) &= 0.09 + 0.2 - (0.09 \times 0.2) = 0.272    
\end{align*}
Thus, WMC computes a leakage probability of 27.2\%, leaving a 2.8\% risk deficit.
This deficit has severe implications for threshold-based policy enforcement.
If an enterprise sets a global safety threshold to block data-sharing actions with a $\geq$ 30\% leakage probability, the robust optimization approach blocks the action, whereas a reference monitor using WMC incorrectly permits it.
Because empirical marginals are often imperfectly calibrated, distributionally-robust optimization ensures soundness by bounding this worst-case risk.

This divergence compounds significantly as agents operate over more files and tool calls. 
In larger derivation graphs, WMC can severely misestimate the true risk depending on the logical structure.
For instance, in a long chain of conjunctions ($x_1 \land \dots \land x_n$) where the underlying events are actually perfectly correlated, the true risk remains high and is bounded only by the minimum marginal probability.
However, WMC simply multiplies these variables, causing its risk estimate to approach zero and leaving a potentially massive vulnerability undetected.

\section{Relaxation of State Measures Optimization}

\subsection{Semidefinite Programming Formulation}

The exact formulation presented in \autoref{sec:exact-optimization-measures} is an exponentially sized linear program (LP) requiring optimization over $2^{|V_F|}$ distinct state configurations, rendering it computationally intractable for real-time execution monitoring. To overcome this complexity, we upper-bound this LP with a polynomially-sized semidefinite program (SDP) by tracking only up to the second-order moments of the state distribution.

Invoking \autoref{thm:Arity2}, every state transition in our compiled derivation graph has an arity of at most two. This structural property guarantees that the graph logic can be completely projected onto first- and second-order expectations. In the exact LP, the structural constraint forces the joint measure $\mu$ to have zero support on states that violate the Datalog transition rules. Because the logical identity $x_v = \phi_v(x_{\mathcal{P}(v)})$ holds with probability 1 across all valid states, taking the expectation of both sides yields the identity $\mathbb{E}[x_v] = \mathbb{E}[\phi_v(x_{\mathcal{P}(v)})]$. Thus, the structural logic of the derivation graph is preserved exactly in moment space as linear equalities.
\begin{subequations}
\label{eq:Exact-Expectation}
\begin{align}
    R^\star = \quad & \max_{\mu} \quad \mathbb{E}[x_t] \\
    \text{s.t.} \quad & \text{ Feasibility constraints:} \nonumber\\
    & \mathbb{E}[x_v] = \mathbb{E}[\phi_v(x_{\mathcal{P}(v)})] \quad \forall v \in V_F : \mathcal{P}(v) \neq \emptyset \label{eq:Exact-Expectation-f}
\\
    & \mathbb{E}[x_v] \in [\ell_v, u_v] \quad \forall v \in V_F : \mathcal{P}(v) = \emptyset \label{eq:Exact-Expectation-a}
\\
    & \text{ Correlation constraints:} \nonumber \\
    & \mathbb{E}[x_v x_{v'}] - \mathbb{E}[x_v]\mathbb{E}[x_{v'}] \geq 0 \quad \forall (v, v') \in \text{POS} \label{eq:Exact-Expectation-POS}\\
    & \mathbb{E}[x_v x_{v'}] - \mathbb{E}[x_v]\mathbb{E}[x_{v'}] \leq 0 \quad \forall (v, v') \in \text{NEG} \label{eq:Exact-Expectation-NEG}\\
    & \mathbb{E}[x_v x_{v'}] - \mathbb{E}[x_v]\mathbb{E}[x_{v'}] = 0 \quad \forall (v, v') \in \text{IND} \label{eq:Exact-Expectation-IND}
\end{align}
\end{subequations}
where all expectations are evaluated directly with respect to $\mu$. While \autoref{eq:Exact-Expectation} reduces the scope of the constraints to low-order moments, searching for a valid global measure $\mu$ remains NP-hard due to the non-linear, non-convex quadratic terms $\mathbb{E}[x_v]\mathbb{E}[x_{v'}]$ in the covariance definitions. Note that we also add the correlation class constraints described in \autoref{sec:probabilistic-datalog}.

To yield a convex, polynomial-time optimization problem, we drop the requirement of maintaining the full underlying distribution $\mu$. We introduce a symmetric, positive semidefinite moment matrix $\mathbf{M} \in \mathbb{R}^{(|V_F|+1) \times (|V_F|+1)}$ to act as a free decision variable representing these low-order expectations directly:
$$\mathbf{M} = \mathbb{E} \left[ \begin{pmatrix} 1 \\ \mathbf{x} \end{pmatrix} \begin{pmatrix} 1 \\ \mathbf{x} \end{pmatrix}^\top \right] \succeq 0$$  
The matrix is indexed from $0$ to $|V_F|$, where the first row and column capture individual marginal probabilities ($M_{0,v} = \mathbb{E}[x_v]$), and the remaining off-diagonal entries capture pairwise joint probabilities ($M_{v,v'} = \mathbb{E}[x_v x_{v'}]$). 

To linearize the non-convex quadratic products $M_{0,v}M_{0,v'} = \mathbb{E}[x_v]\mathbb{E}[x_{v'}]$, we introduce a sparse set of scalar auxiliary variables $T_{v,v'}$ bounded tightly by continuous McCormick envelopes~\citep{mccormick1976computability} defined over the marginal intervals $[\ell_v, u_v]$. We formalize the complete semidefinite programming relaxation as follows:

\begin{subequations}
\label{eq:SDP-formulation}
\begin{align}
    R^\star_{\text{SDP}} = \quad & \max_{\mathbf{M}, \mathbf{T}} \quad M_{0,t} \\
    \text{s.t.} \quad & \text{Feasibility constraints:} \nonumber\\
    & M_{0,v} = \mathbb{E}_{\mathbf{M}}\left[\phi_v (\mathcal{P}(v))\right] \quad \forall v \in V_F : \mathcal{P}(v) \neq \emptyset \label{eq:SDP-f}\\
    & \ell_v \leq M_{0,v} \leq u_v \quad \forall v \in V_F : \mathcal{P}(v) = \emptyset \label{eq:SDP-a}\\
    & \text{Correlation constraints:} \nonumber\\
    & \text{Let } C = \text{POS} \cup \text{NEG} \cup \text{IND} \nonumber\\
    & M_{v,v'} - T_{v,v'} \geq 0 \quad \forall (v,v') \in \text{POS} \label{eq:SDP-b}\\
    & M_{v,v'} - T_{v,v'} \leq 0 \quad \forall (v,v') \in \text{NEG} \label{eq:SDP-c}\\
    & M_{v,v'} - T_{v,v'} = 0 \quad \forall (v,v') \in \text{IND} \label{eq:SDP-d}\\
    & T_{v,v'} \geq M_{0,v}\ell_{v'} + M_{0,v'}\ell_{v} - \ell_v \ell_{v'} \quad \forall (v,v') \in C \label{eq:SDP-mcc1}\\
    & T_{v,v'} \geq M_{0,v}u_{v'} + M_{0,v'}u_{v} - u_v u_{v'} \quad \forall (v,v') \in C \label{eq:SDP-mcc2}\\
    & T_{v,v'} \leq M_{0,v}\ell_{v'} + M_{0,v'}u_{v} - u_v \ell_{v'} \quad \forall (v,v') \in C \label{eq:SDP-mcc3}\\
    & T_{v,v'} \leq M_{0,v}u_{v'} + M_{0,v'}\ell_{v} - \ell_v u_{v'} \quad \forall (v,v') \in C \label{eq:SDP-mcc4}\\
   &  \text{Moment feasibility constraints:} \nonumber\\
    & M_{v,v} = M_{0,v} \quad \forall v \in \{1, \dots, |V_F|\} \label{eq:SDP-moment-diag}\\
    & M_{0,0} = 1, \quad \mathbf{M} \succeq 0 \label{eq:SDP-psd}
\end{align}
\end{subequations}

For any derived fact $v$ with a binary parent set $\mathcal{P}(v) = \{u, w\}$, the expectation operator $\mathbb{E}_{\mathbf{M}}\left[\phi_v (\mathcal{P}(v))\right]$ in constraint \eqref{eq:SDP-f} is a strictly linear combination of the elements of $\mathbf{M}$:
\begin{align*}
\mathbb{E}_{\mathbf{M}}\left[\phi_v (u, w)\right] = \,\, & M_{u,w}\phi_v(1, 1) \\
+ \, & (1 - M_{0,u} - M_{0,w} + M_{u,w})\phi_v(0, 0) \\
+ \, & (M_{0,u} - M_{u,w})\phi_v(1, 0) \\
+ \, & (M_{0,w} - M_{u,w})\phi_v(0, 1)
\end{align*}

The resulting optimization problem contains a single matrix variable $\mathbf{M}$ with $\frac{1}{2}(|V_F|+1)(|V_F|+2)$ unique entries, and at most $|\mathcal{C}|$ scalar auxiliary variables for $\mathbf{T}$. The number of constraints scales linearly with the number of derived facts $|V_F|$ and the number of inferred correlation pairs $|\mathcal{C}|$. Consequently, the total size of the optimization problem is $O(|V_F|^2 + |\mathcal{C}|)$, rendering it solvable in polynomial time via standard interior-point methods.

\subsection{Soundness of Semidefinite Relaxation}

To anchor our framework within the ideal runtime guarantees of a probabilistic reference monitor (\autoref{sec:probabilistic-execution-monitoring}), we must prove that our relaxation yields a sound risk estimate.
This soundness ensures that the reference monitor will never under-block a policy violation, ensuring strict security guarantees even if conservative computed bounds occasionally induce over-blocking and utility degradation.

\begin{theorem}[Soundness of Semidefinite Relaxation]
Let $R^*$ be the optimal value of the exact optimization problem over joint state measures~(\ref{eq:exact-optimization}), and let $R_{\text{SDP}}^*$ be the optimal value of the relaxed semidefinite program~(\ref{eq:SDP-formulation}). The semidefinite program yields a conservative over-approximation of the true execution risk:
$$R^* \leq R_{\text{SDP}}^*$$
\end{theorem}

\begin{proof}

We prove soundness by showing that any joint probability distribution $\mu$ that is feasible in the exact linear program \eqref{eq:exact-optimization} maps to a feasible assignment for the decision variables $(\mathbf{M}, \mathbf{T})$ in the semidefinite program \eqref{eq:SDP-formulation}. Start with \eqref{eq:exact-optimization} and substitute $M_{0,v} \leftarrow \mathbb{E}_\mu[x_v]$, $M_{v,v'} \leftarrow \mathbb{E}_\mu[x_v x_{v'}]$, $T_{v,v'} \leftarrow \mathbb{E}_\mu[x_v]\mathbb{E}_\mu[x_{v'}]$. Then, there is a clear correspondence between equations \eqref{eq:Exact-Expectation-a} and \eqref{eq:SDP-a} and similarly for all constraints in \eqref{eq:Exact-Expectation} there is a corresponding constraint in \eqref{eq:SDP-formulation}. Only two constraints involve some nontrivial mapping: The expectation operator is linear, taking the expectation on both sides of the exact structural logic $x_v = \phi_v(x_u, x_w)$ maps perfectly to the linear identity in constraint \eqref{eq:SDP-f}. Finally, since the marginal expectations are bounded by $[\ell_v, u_v]$, the products of these expectations structurally satisfy the bilinear McCormick inequalities (\ref{eq:SDP-mcc1}--\ref{eq:SDP-mcc4}) by definition. 

Thus, if $\mathbf{M}$ is a valid matrix of second order moments under a true probability measure $\mu$ that satisfies the constraints of \eqref{eq:Exact-Expectation}, then it must satisfy the constraints (\ref{eq:SDP-f}--\ref{eq:SDP-mcc4}). Furthermore, since $x^2=x$ for any binary variable $x$, it must also satisfy $\mathbf{M}_{v,v}=\mathbf{M}_{0,v}, M_{00}=1, \mathbf{M}\succeq 0$.

Thus, for any feasible solution $\mu$ to \eqref{eq:Exact-Expectation}, constructing $\mathbf{M}$ to be the matrix of second order moments produces a feasible solution to \eqref{eq:SDP-formulation} with the exact same objective value. However not all solutions $\mathbf{M}$ that are feasible for \eqref{eq:SDP-formulation} are guaranteed to be valid second order moments for a distribution $\mu$ that satisfies \eqref{eq:Exact-Expectation}. Hence the optimal value of \eqref{eq:SDP-formulation} is a guaranteed upper bound on the optimal value of \eqref{eq:Exact-Expectation}.
\end{proof}

\section{Evaluation}

In this section, we evaluate our proposed SDP optimization relaxation against baselines for deterministic and probabilistic verification.
We aim to answer the following research questions:

\begin{description}
\item[\textbf{RQ1}] How effectively does our SDP relaxation balance system utility and security guarantees on terminal agent benchmarks compared to prior work? \emph{We find that the SDP achieves the optimal tradeoff outperforming prior art on utility at various fixed security levels (\autoref{tab:performance_extended}).
} \\

\item[\textbf{RQ2}] How severely do strict independence assumptions underestimate security risks as input correlations increase? \emph{We find that assuming independence can systematically overestimate risks when probabilistic predicates are known to be strictly positively correlated, and the overestimation increases as correlations increase (\autoref{tab:intercode_correlations})}. \\

\item[\textbf{RQ3}] Can we match the performance of prior art on general probabilistic reasoning in security policies written in Datalog? \emph{We find that SDP relaxation is effective even on classical security policies matching prior state of the art \citep{wang2025praline}, albeit at a higher average latency (\autoref{tab:performance_extended}, \autoref{tab:latency}}). \\

\end{description}

The evaluation results, detailed in subsequent sections, showed that our relaxation demonstrates scalability and tightness on security policy tasks with defined input correlations. We also showed that our sound SDP relaxation matches the security of the robust baseline and attains high utility. Finally, we observed that probabilistic inference under independence results in utility degradation under certain correlations.

\subsection{Benchmarks}

We evaluated our framework across a diverse suite of agentic safety and probabilistic logic programming tasks designed to test scalability, multi-turn security enforcement, and resilience to correlated inputs.

\textbf{Intercode-NL2Bash.}
We considered 197 trajectories from the Intercode-NL2Bash benchmark~\citep{yang2023intercode,lin2018nl2bash}.
To compile these linear shell commands into an explicit lineage, we materialized standard shell pipes (\texttt{|}) into intermediate file reads and writes, and we mapped each bash command to one of our taint transition types (\autoref{tab:taint-semantics}).
Because Intercode trajectories lack explicit data exfiltration events, we appended a terminal data-sharing action to the final output of each trace to evaluate policy compliance.

\textbf{ATBench.}
We also considered 377 agent trajectories from ATBench~\citep{liu2026agentdogdiagnosticguardrailframework}.
Following the same compilation workflow as Intercode, we mapped interactive tool invocations directly to our taint transition semantics.
Unlike Intercode, ATBench natively includes explicit data transmission actions across trust boundaries, which we directly leveraged to evaluate policy violations.

\textbf{Side-channel vulnerability analysis.}
To test our relaxation on other security policies written in Datalog, we considered benchmark tasks from Praline~\citep{wang2025praline}.
We considered six side-channel analysis tasks~\citep{wang2019sidechannel}, which model information-flow vulnerabilities introduced by compiler register allocation in cryptographic software, tracking whether sensitive intermediate computation results inadvertently share CPU registers or memory locations.
These tasks contain up to eight queries and involve complex, highly overlapping dependency chains in their derivation graphs.
Unlike Intercode and ATBench, upper and lower prior marginal bounds for these programs are natively provided by the benchmark suite.

\subsection{Experimental Setup}\label{experimental-setup}

To compute derivation graphs for each task, we employed Praline's instrumentation of the Soufflé Datalog evaluation engine \citep{jordan2016souffle}.
While standard Soufflé suppresses duplicate derivation paths for the same relation, this instrumentation allows us to extract complete derivation graphs, capturing all parallel rules that can be applied to derive a given fact.

While our goal is distributional robustness, we utilized correlations when they are known in advance.
This knowledge could exist in agentic settings if a user provides hints, e.g., files in the same directory might share sensitive data, so their PII detector outputs would be positively correlated.
Therefore, we adopted Praline's two-phase correlation inference algorithm to discover positive, negative, or independent relationships between predicates.
In the first phase, the algorithm inferred statistical correlations among input facts in the same correlation class by checking if their joint probabilities deviate from the product of their individual marginals.
The second phase performed a lightweight static analysis, propagating known correlation types forward to output predicates.

With this inference algorithm, we were able to apply known correlation constraints to the SDP relaxation for the side channel tasks, where correlation classes are defined in advance.
For these tasks, prior marginal probabilities for input facts are already provided by the benchmark. 
In contrast, we assumed no specific correlations for tasks in Intercode and ATBench (unless otherwise specified), and we randomly generated classifier bounds $[\ell_v, u_v]$ and redaction failure probabilities to simulate environmental uncertainty.

For all terminal agent tasks (Intercode and ATBench), we computed the ground truth upper bound using the exact optimization linear program (\autoref{eq:exact-optimization}).
Because exact optimization is infeasible for complex derivation graphs, we excluded tasks whose combined number of input facts and rules exceeds 12.
To evaluate our framework, we used the Splitting Conic Solver \citep{Odonoghue2016scs} to solve the SDP.

Inline enforcement requires efficient trajectory evaluation to prevent execution stalls; excessive latency defeats the purpose of autonomy by forcing frequent user intervention.
Therefore, we imposed a 15-second timeout on all solver-based verifiers. If an engine failed to converge within this window, it fell back to its looser bounds (when applicable).

\subsection{Baselines}\label{sec:baselines}

\begin{table*}[ht] 
\centering
\small
\caption{Comparison of model performance and AUC across Intercode, ATBench, and Side Channel benchmarks under varying global security risk thresholds $\tau$. Dashes indicate that there were no positive samples for a given threshold.}
\resizebox{\textwidth}{!}{
\begin{tabular}{llccccccccc}
\toprule
\multirow{2}{*}{\textbf{Threshold}} & \multirow{2}{*}{\textbf{Model}} & \multicolumn{3}{c}{\textbf{Intercode}} & \multicolumn{3}{c}{\textbf{ATBench}} & \multicolumn{3}{c}{\textbf{Side Channel}} \\
\cmidrule(lr){3-5} \cmidrule(lr){6-8} \cmidrule(lr){9-11}
& & Utility & Security & AUC & Utility & Security & AUC & Utility & Security & AUC \\
\midrule
\multirow{4}{*}{\shortstack[l]{\textbf{High} \\ \textbf{Security}\\\textbf{($\tau = 0.2$)}}} 
& Praline & $0.979$ & $\mathbf{1.000}$ & $0.980$ & $\mathbf{1.000}$ & $\mathbf{1.000}$ & $\mathbf{1.000}$ & $\mathbf{1.000}$ & $\mathbf{1.000}$ & --- \\
& Monte Carlo & $\mathbf{1.000}$ & $0.600$ & $0.991$ & $\mathbf{1.000}$ & $0.884$ & $\mathbf{1.000}$ & $\mathbf{1.000}$ & $\mathbf{1.000}$ & --- \\
& Deterministic & $0.986$ & $0.716$ & $0.920$ & $\mathbf{1.000}$ & $0.835$ & $0.986$ & $\mathbf{1.000}$ & $\mathbf{1.000}$ & --- \\
\cmidrule(lr){2-11}
& SDP (Ours) & $\mathbf{1.000}$ & $\mathbf{1.000}$ & $\mathbf{0.998}$ & $\mathbf{1.000}$ & $\mathbf{1.000}$ & $\mathbf{1.000}$ & $\mathbf{1.000}$ & $\mathbf{1.000}$ & --- \\
\midrule
\multirow{4}{*}{\shortstack[l]{\textbf{Medium} \\ \textbf{Security}\\\textbf{($\tau = 0.5$)}}} 
& Praline & $0.979$ & $\mathbf{1.000}$ & $0.975$ & $\mathbf{1.000}$ & $\mathbf{1.000}$ & $\mathbf{1.000}$ & $\mathbf{1.000}$ & $\mathbf{1.000}$ & $\mathbf{1.000}$ \\
& Monte Carlo & $\mathbf{1.000}$ & $0.702$ & $0.989$ & $\mathbf{1.000}$ & $0.919$ & $\mathbf{1.000}$ & $\mathbf{1.000}$ & $0.810$ & $\mathbf{1.000}$ \\
& Deterministic& $\mathbf{1.000}$ & $0.702$ & $0.887$ & $\mathbf{1.000}$ & $0.950$ & $0.993$ & $\mathbf{1.000}$ & $0.143$ & $0.905$ \\
\cmidrule(lr){2-11}
& SDP (Ours) & $\mathbf{1.000}$ & $\mathbf{1.000}$ & $\mathbf{0.997}$ & $\mathbf{1.000}$ & $\mathbf{1.000}$ & $\mathbf{1.000}$ & $\mathbf{1.000}$ & $\mathbf{1.000}$ & $\mathbf{1.000}$ \\
\midrule
\multirow{4}{*}{\shortstack[l]{\textbf{Low}\\\textbf{Security}\\\textbf{($\tau = 0.8$)}}} 
& Praline & $0.971$ & $\mathbf{1.000}$ & $0.972$ & $0.996$ & $\mathbf{1.000}$ & $\mathbf{0.999}$ & $\mathbf{1.000}$ & $\mathbf{1.000}$ & $\mathbf{1.000}$ \\
& Monte Carlo & $\mathbf{1.000}$ & $0.636$ & $\mathbf{0.991}$ & $\mathbf{1.000}$ & $0.787$ & $0.993$ & $\mathbf{1.000}$ & $\mathbf{1.000}$ & $\mathbf{1.000}$ \\
& Deterministic & $0.969$ & $0.939$ & $0.959$ & $0.949$ & $0.965$ & $0.975$ & $0.222$ & $\mathbf{1.000}$ & $0.611$ \\
\cmidrule(lr){2-11}
& SDP (Ours) & $\mathbf{1.000}$ & $\mathbf{1.000}$  & $\mathbf{0.991}$ & $0.983$ & $\mathbf{1.000}$ & $0.998$ & $\mathbf{1.000}$ & $\mathbf{1.000}$ & $\mathbf{1.000}$ \\
\bottomrule
\end{tabular}
}
\label{tab:performance_extended}
\end{table*}

In our evaluation, we considered two probabilistic verification engines as baselines, as well as a deterministic engine with local thresholds applied to probabilistic predicates.

\textbf{Praline.} Praline \citep{wang2025praline} is an inference engine designed for Datalog programs with correlated inputs. 
It first executes the two-phase correlation analysis described in \autoref{experimental-setup}.
Praline then performs approximate inference through the derivation graph using constrained optimization, producing sound initial probability bounds.
To tighten these intervals, Praline employs an binary search refinement procedure.
This loop repeatedly invokes an SMT solver to check constraint satisfiability until converging on a target that is within $\delta$ of the true bounds.
We evaluated Praline using the Solving Constraint Integer Programs (SCIP) solver \citep{scip2008} for constrained optimization and Z3 for SMT solving \citep{demouraz32008}.

\textbf{Monte Carlo.} We implemented a Monte Carlo sampling routine to approximate the trajectory risk under the assumption of strict independence.
For each base fact $v$, we drew a point probability $p_v \sim \text{Uniform}(\ell_v, u_v)$ and instantiated its concrete truth value via a Bernoulli trial $x_v \sim \text{Bernoulli}(p_v)$.
We ran the deterministic Datalog program on these sampled values to evaluate query satisfaction, averaging the final executing risk across $10{,}000$ simulations.

\textbf{Deterministic engine.}
To simulate a deterministic reference monitor, we varied a local binarization threshold from $0.1$ to $0.9$ in increments of $0.1$.
For each local threshold, continuous input probabilities were converted to binary assignments and executed directly through the deterministic Datalog program.

\subsection{RQ1: Security-Utility Trade-Off for Terminal Agent Settings}

To evaluate the security-utility balance across the different engines, we modeled policy enforcement as a binary classification task where a positive label corresponds to a reference monitor blocking an agent's action.
Under this formulation, a False Negative (FN) represents a security failure where a policy-violating trajectory is incorrectly allowed, whereas a False Positive (FP) represents a utility loss where a safe trajectory is unnecessarily blocked.
True Positives (TP) and True Negatives (TN) denote correctly blocked violations and allowed safe actions, respectively.
We formalize \emph{utility} as precision, where $\text{precision} = \frac{\text{TP}}{\text{TP} + \text{FP}}$, which measures the proportion of blocked actions that are genuine security policy violations.
Conversely, we formalize \emph{security} as recall, where $\text{recall} = \frac{\text{TP}}{\text{TP} + \text{FN}}$, which measures the proportion of unsafe actions that were correctly blocked.

\begin{table}[t]
\centering
\normalsize
\caption{Average latency comparison (ms) between Praline and SDP across Intercode, ATBench, and Side Channel.}
\begin{tabular}{lcc}
\toprule
\textbf{Benchmark} & \textbf{Praline} & \textbf{SDP (Ours)} \\
\midrule
Intercode & $1{,}015$ & $\mathbf{221}$ \\
ATBench & $7{,}227$ & $\mathbf{303}$ \\
Side Channel & $\mathbf{330}$ & $1{,}927$ \\
\bottomrule
\end{tabular}
\label{tab:latency}
\end{table}

We evaluated the security-utility trade-off across the benchmarks by sweeping the global safety threshold $\tau$.
We report results for utility, security, and area under the utility-security curve (AUC) in \autoref{tab:performance_extended}.
For each benchmark and model, we performed one evaluation run, only changing the ground truth label depending on the given global threshold.
Therefore, we only report one set of latencies for each benchmark (\autoref{tab:latency}).
For the deterministic engine, we report utility and security for the local threshold that achieved the highest AUC for each benchmark and global safety threshold.

Since Praline produces a sound over-estimate of the actual risk, it achieves perfect security; however, solver timeouts on several tasks forced it to return loose, unrefined bounds, resulting in unnecessary blocking and degraded utility.
Conversely, MC achieves perfect utility because its independence assumption consistently underestimates worst-case risk, ensuring it never over-blocks at the cost of compromised security.
The deterministic engine discards probabilistic information entirely, often leading to low security, utility, or both.
In fact, the top-performing deterministic engine displays a degradation in utility as the global threshold increases and a degradation in security as the global threshold decreases.
Our SDP relaxation effectively balances this trade-off by computing sound, tight bounds that maximize both security and utility without solver timeouts, thereby maintaining low enforcement computational overhead.

\input{figures/pr_curves}

To visualize this balance across verification engines, we plotted the utility-security (precision-recall) curves for the Intercode tasks for $\tau \in \{0.2, 0.5, 0.8\}$ (\autoref{fig:pr_roc_curves}).
The curves demonstrate that for high and medium global security thresholds, the deterministic baseline suffers from poor utility even at relatively low security levels, rendering it ineffective for practical policy enforcement. In contrast, Praline, Monte Carlo, and our SDP relaxation all demonstrate high utility at high security levels across all global thresholds. However, as the security constraints tighten, Praline and Monte Carlo generally exhibit an earlier drop in utility than SDP. This gap is reflected in the overall AUC scores, where our SDP relaxation consistently outperforms the baselines.

\subsection{RQ2: Independence Assumptions}

For \textbf{RQ1}, we concluded that Monte Carlo consistently underestimates the upper bound on the violation probability due to its independence assumption, resulting in perfect utility but lower security.
For \textbf{RQ2}, we aim to understand how independence assumptions can reduce utility under certain correlation assumptions.
Accordingly, we created three versions of the Intercode tasks, each with a different level of correlation among the inputs.
Under no-correlation assumptions, the setting is the same as in \textbf{RQ1}.
Under medium-correlation assumptions, all redactions were declared as independent, and the input files were randomly divided into two correlation classes, indicating that the files within each class were positively correlated.
Under high-correlation assumptions, all redactions were declared as independent, and all input files were assigned to the same correlation class, indicating that all input files were positively correlated.
To compute the ground truth upper bound on the risk probability, we modified the exact optimization LP (\autoref{sec:exact-optimization-measures}) to impose these independent and positive correlation constraints.

In the most general verification setting, an engine that uses Monte Carlo sampling cannot make any assumptions about whether choosing its samples, either using the lower bound $\ell_v$ or upper bound $u_v$ for an input fact $v$, will result in the greatest violation risk.
However, in this specific information-flow setting, state transitions are monotonic, meaning that considering the full range of marginal probabilities $[\ell_v, u_v]$ does not change the result of the exact LP optimization. 
Therefore, we modified Monte Carlo to instantiate the concrete truth value for each input fact $v$ via a Bernoulli trial over its upper bound, i.e., $x_v \sim \text{Bernoulli}(u_v)$, rather than first drawing a point probability from $\text{Uniform}(\ell_v, u_v)$.

We report results for utility and security under each of the three correlation settings (None, Medium, and High) and the three global security risk thresholds $\tau \in \{0.2, 0.5, 0.8\}$ (\autoref{tab:intercode_correlations}).
The results under no correlation assumptions reflect what we observed in \textbf{RQ1}, where Monte Carlo's independence assumptions caused it to consistently underestimate the worst-case risk.
This led Monte Carlo to predict many false negatives (high utility but lower security).
However, under medium-correlation assumptions, Monte Carlo experienced a drop in utility that degraded even more as inputs became highly correlated, particularly under a high-security global threshold.
Without explicit constraints, the worst-case risk was driven by a joint distribution where base input facts were independent, though internal redaction failures remained perfectly correlated.
Because Monte Carlo enforces global independence across all variables, it approximated this unconstrained input profile closely enough to reduce false positives.
However, when medium- or high-input correlations were introduced, the true worst-case risk bound shrunk.
While the SDP relaxation integrates these input constraints to optimize utility, Monte Carlo's rigid independence assumption overestimates the true risk, resulting in excessive false positives.

\begin{table}[t]
\centering
\small 
\caption{Intercode utility and security for Monte Carlo (MC) and SDP under different global security risk thresholds $\tau$ and input correlation assumptions.}
\label{tab:intercode_correlations}
\begin{tabular}{lllcc}
\toprule
\textbf{Threshold} & \textbf{Correlation} & \textbf{Model} & \textbf{Utility} & \textbf{Security} \\
\midrule
\multirow{6}{*}{\shortstack[l]{\textbf{High} \\ \textbf{Security} \\ \textbf{($\tau = 0.2$)}}} 
& \multirow{2}{*}{None} & MC & $0.968$ & $0.958$ \\
&                       & SDP   & $\mathbf{1.000}$ & $\mathbf{1.000}$ \\
\cmidrule(lr){2-5}
& \multirow{2}{*}{Medium} & MC & $0.957$ & $\mathbf{1.000}$ \\
&                         & SDP  & $\mathbf{1.000}$ & $\mathbf{1.000}$ \\
\cmidrule(lr){2-5}
& \multirow{2}{*}{High} & MC & $0.606$ & $\mathbf{1.000}$ \\
&                          & SDP   & $\mathbf{0.982}$ & $\mathbf{1.000}$ \\

\midrule

\multirow{6}{*}{\shortstack[l]{\textbf{Medium} \\ \textbf{Security} \\ \textbf{($\tau = 0.5$)}}} 
& \multirow{2}{*}{None} & MC & $\mathbf{1.000}$ & $0.830$ \\
&                       & SDP   & $\mathbf{1.000}$ & $\mathbf{1.000}$ \\
\cmidrule(lr){2-5}
& \multirow{2}{*}{Medium} & MC & $0.950$ & $\mathbf{1.000}$ \\
&                         & SDP   & $\mathbf{1.000}$ & $\mathbf{1.000}$ \\
\cmidrule(lr){2-5}
& \multirow{2}{*}{High} & MC & $0.833$ & $\mathbf{1.000}$ \\
&                          & SDP   & $\mathbf{1.000}$ & $\mathbf{1.000}$ \\

\midrule

\multirow{6}{*}{\shortstack[l]{\textbf{Low} \\ \textbf{Security} \\ \textbf{($\tau = 0.8$)}}} 
& \multirow{2}{*}{None} & MC & $0.969$ & $0.939$ \\
&                       & SDP   & $\mathbf{1.000}$ & $\mathbf{1.000}$ \\
\cmidrule(lr){2-5}
& \multirow{2}{*}{Medium} & MC & $0.939$ & $\mathbf{1.000}$ \\
&                         & SDP   & $\mathbf{1.000}$ & $\mathbf{1.000}$ \\
\cmidrule(lr){2-5}
& \multirow{2}{*}{High} & MC & $0.853$ & $\mathbf{1.000}$ \\
&                          & SDP   & $\mathbf{1.000}$ & $\mathbf{1.000}$ \\
\bottomrule
\end{tabular}
\end{table}

\subsection{RQ3: General Probabilistic Security Policies}

For \textbf{RQ3}, we evaluated the performance of our SDP relaxation on the side-channel vulnerability analysis tasks the Praline benchmarks.
These tasks provide both upper and lower prior marginal bounds, as well as defined correlation classes.
For these tasks, we adopted the Praline interpretation of input correlation classes.
Specifically, Praline treats inputs in the same correlation class as ``maybe correlated,'' meaning they could be positively correlated or independent.
This is in contrast to our interpretation in \textbf{RQ2}, where inputs defined in the same correlation class were treated as strictly positively correlated.
This is an important distinction, since risk may be maximized with the independent assumption or the correlation assumption, depending on the policy.

Since these tasks contain more derivation graph nodes than Intercode and ATBench, computing the exact LP optimization is intractable.
Therefore, we used the $\delta$-precise upper bounds computed by Praline running \emph{without an execution timeout} to serve as our exact ground-truth oracle.
Our results in \autoref{tab:performance_extended} demonstrate that SDP and Praline attain perfect security and utility across all global thresholds.
Conversely, the independent Monte Carlo baseline and the deterministic engine demonstrated degraded performance across both metrics.
While the introduction of correlation constraints over these complex side-channel graphs caused our SDP framework to experience a higher average latency than Praline (\autoref{tab:latency}), the SDP relaxation avoids the catastrophic worst-case exponential tail latencies and solver timeouts inherent to exact SMT-based reasoning, demonstrating its viability as a predictable, polynomial-time verification engine.

\section{Limitations and Future Work}

\textbf{Trajectory horizon and bound degradation.} The primary limitation of our approach concerns verification over extended execution horizons.
As an agent's trajectory scales, distributionally robust optimization over deep execution traces with high degrees of input-fact merging can cause the risk upper bound to converge to $1.0$.
This phenomenon occurs even if the initial marginal probabilities are minimal.
When environmental correlations are unknown \emph{a priori}, computing a sound upper bound requires conservatively accounting for the worst-case valid joint distribution across every state transition. 
While this bound degradation is an inherent constraint of distribution-free optimization rather than a flaw of our specific relaxation, it implies that the precision of our framework is limited by the quality of the marginals and specified correlation classes.

\textbf{Computational scaling.} Our evaluation demonstrates that as derivation graphs deepen and more correlation constraints are integrated, the matrix dimensions of the SDP relaxation expand. 
This expansion introduces a polynomial increase in the number of decision variables and linear constraints, resulting in higher solver latencies during verification.
Beyond a certain graph complexity, this solver overhead could become prohibitively expensive for real-time policy enforcement.
While an operational reference monitor could mitigate this by resorting to a manual user-approval prompt when the optimization problem encounters scalability limits, this fallback mechanism risks inducing approval fatigue and fundamentally undermines the utility of autonomous agent execution.

\textbf{Ambiguous tool semantics.} Another limitation of this work is its reliance on pre-defined tool semantics.
While many bash commands can be modeled easily (e.g., \texttt{cp} propagates taint from the source file to the destination file and \texttt{touch} creates a clean file, in their simplest invocations), for many tools these semantics are unclear.
Even the Intercode tasks, which represent some of the simplest trajectories among terminal agent benchmarks, contain commands with ambiguous semantics due to their use of execution wrappers around sub-commands, e.g., \texttt{xargs}.
While we manually modeled such commands in our evaluation, such modeling is difficult to implement in practice and runs the risk of allowing mimicry attacks~\citep{10.1145/586110.586145} when tool semantics are underspecified.
First, agents are not usually restricted to a small set of bash commands like Intercode.
Real-world agents typically write and execute arbitrary Python or Bash scripts on-the-fly, making them impossible to model in advance \citep{jimenez2024swebench,merrill2026terminalbench}.
Second, practical deployments may require checking safety specifications beyond strict information-flow tracking---there could be a need to enforce many different agent policies, each requiring different defined transitions.
Therefore, one area for future work is to automatically model commands and scripts as they are generated by the agent, possibly incorporating uncertainty about inferred specifications into the verification itself.

\section{Conclusion}

As autonomous AI agents are increasingly trusted with access to sensitive data, deterministic security policy enforcement fails to manage the inherent ambiguity of real-world environments. 
In this paper, we propose a paradigm shift toward probabilistic agent verification, modeling multi-step trajectory state transitions via a distributionally robust Probabilistic Datalog framework.
We formalize this verification as an exact optimization problem that computes a sound upper bound on policy violation risks without relying on unsafe independence assumptions. 
To ensure runtime viability, we introduce a polynomial-time SDP relaxation that tracks second-order moments of the joint probability distribution. 
Our evaluation across terminal and tool-calling benchmarks demonstrates that our framework bridges the gap between soundness and tractability, maintaining low latency overhead while optimizing the security-utility trade-off. 
Ultimately, our framework lays the groundwork for distributionally robust agentic verification, providing a scalable foundation for securing autonomous agents operating in ambiguous environments.

\bibliography{references}

@article{wang2025praline,
author = {Wang, Jingbo and Halalingaiah, Shashin and Chen, Weiyi and Wang, Chao and Dillig, I\c{s}\i{}l},
title = {Probabilistic Inference for Datalog with Correlated Inputs},
year = {2025},
issue_date = {October 2025},
publisher = {Association for Computing Machinery},
address = {New York, NY, USA},
volume = {9},
number = {OOPSLA2},
url = {https://doi.org/10.1145/3763058},
doi = {10.1145/3763058},
journal = {Proc. ACM Program. Lang.},
month = oct,
articleno = {280},
numpages = {28}
}

@inproceedings{
yao2023react,
title={ReAct: Synergizing Reasoning and Acting in Language Models},
author={Shunyu Yao and Jeffrey Zhao and Dian Yu and Nan Du and Izhak Shafran and Karthik R Narasimhan and Yuan Cao},
booktitle={The Eleventh International Conference on Learning Representations },
year={2023},
url={https://openreview.net/forum?id=WE_vluYUL-X}
}

@misc{costa2025securingaiagentsinformationflow,
      title={Securing AI Agents with Information-Flow Control}, 
      author={Manuel Costa and Boris Köpf and Aashish Kolluri and Andrew Paverd and Mark Russinovich and Ahmed Salem and Shruti Tople and Lukas Wutschitz and Santiago Zanella-Béguelin},
      year={2025},
      eprint={2505.23643},
      archivePrefix={arXiv},
      primaryClass={cs.CR},
      url={https://arxiv.org/abs/2505.23643}, 
}

@article{wiesemann2014dro,
author = {Wiesemann, Wolfram and Kuhn, Daniel and Sim, Melvyn},
title = {Distributionally Robust Convex Optimization},
year = {2014},
issue_date = {December 2014},
publisher = {INFORMS},
address = {Linthicum, MD, USA},
volume = {62},
number = {6},
issn = {0030-364X},
journal = {Operations Research},
month = dec,
pages = {1358–1376},
numpages = {19}
}

@inproceedings{yang2023intercode,
author = {Yang, John and Prabhakar, Akshara and Narasimhan, Karthik and Yao, Shunyu},
title = {InterCode: standardizing and benchmarking interactive coding with execution feedback},
year = {2023},
publisher = {Curran Associates Inc.},
address = {Red Hook, NY, USA},
booktitle = {Proceedings of the 37th International Conference on Neural Information Processing Systems},
articleno = {1035},
numpages = {29},
location = {New Orleans, LA, USA},
series = {NIPS '23}
}

@inproceedings{lin2018nl2bash,
    title = "{NL}2{B}ash: A Corpus and Semantic Parser for Natural Language Interface to the Linux Operating System",
    author = "Lin, Xi Victoria  and
      Wang, Chenglong  and
      Zettlemoyer, Luke  and
      Ernst, Michael D.",
    editor = "Calzolari, Nicoletta  and
      Choukri, Khalid  and
      Cieri, Christopher  and
      Declerck, Thierry  and
      Goggi, Sara  and
      Hasida, Koiti  and
      Isahara, Hitoshi  and
      Maegaard, Bente  and
      Mariani, Joseph  and
      Mazo, H{\'e}l{\`e}ne  and
      Moreno, Asuncion  and
      Odijk, Jan  and
      Piperidis, Stelios  and
      Tokunaga, Takenobu",
    booktitle = "Proceedings of the Eleventh International Conference on Language Resources and Evaluation ({LREC} 2018)",
    month = may,
    year = "2018",
    address = "Miyazaki, Japan",
    publisher = "European Language Resources Association (ELRA)",
    url = "https://aclanthology.org/L18-1491/"
}

@misc{wu2024systemleveldefenseindirectprompt,
      title={System-Level Defense against Indirect Prompt Injection Attacks: An Information Flow Control Perspective}, 
      author={Fangzhou Wu and Ethan Cecchetti and Chaowei Xiao},
      year={2024},
      eprint={2409.19091},
      archivePrefix={arXiv},
      primaryClass={cs.CR},
      url={https://arxiv.org/abs/2409.19091}, 
}

@misc{palumbo2026formalpolicyenforcementrealworld,
      title={Formal Policy Enforcement for Real-World Agentic Systems}, 
      author={Nils Palumbo and Sarthak Choudhary and Jihye Choi and Guy Amir and Prasad Chalasani and Somesh Jha},
      year={2026},
      eprint={2602.16708},
      archivePrefix={arXiv},
      primaryClass={cs.CR},
      url={https://arxiv.org/abs/2602.16708}, 
}

@inproceedings{debenedetti2025defeatingpromptinjectionsdesign,
  author = {Debenedetti, Edoardo and Shumailov, Ilia and Fan, Tianqi and Hayes, Jamie and Carlini, Nicholas and Fabian, Daniel and Kern, Christoph and Shi, Chongyang and Terzis, Andreas and Tram\`er, Florian},
  booktitle = {4th IEEE Conference on Secure and Trustworthy Machine Learning},
  title = {Defeating Prompt Injections by Design},
  url = {https://arxiv.org/abs/2503.18813},
  year = {2026}
}

@inproceedings{
chen2025shieldagent,
title={ShieldAgent: Shielding Agents via Verifiable Safety Policy Reasoning},
author={Zhaorun Chen and Mintong Kang and Bo Li},
booktitle={Forty-second International Conference on Machine Learning},
year={2025},
url={https://openreview.net/forum?id=DkRYImuQA9}
}

@inproceedings{wang2026agentspec,
author = {Haoyu Wang and Chris Poskitt and Jun Sun },
booktitle = { 2026 IEEE/ACM 48th International Conference on Software Engineering (ICSE) },
title = {{ AgentSpec: Customizable Runtime Enforcement for Safe and Reliable LLM Agents }},
year = {2026},
volume = {},
ISSN = {},
publisher = {IEEE Computer Society},
address = {Los Alamitos, CA, USA},
month =May}

@misc{kamath2025enforcingtemporalconstraintsllm,
      title={Enforcing Temporal Constraints for LLM Agents}, 
      author={Adharsh Kamath and Sishen Zhang and Calvin Xu and Shubham Ugare and Gagandeep Singh and Sasa Misailovic},
      year={2025},
      eprint={2512.23738},
      archivePrefix={arXiv},
      primaryClass={cs.PL},
      url={https://arxiv.org/abs/2512.23738}, 
}

@misc{miculicich2025veriguardenhancingllmagent,
      title={VeriGuard: Enhancing LLM Agent Safety via Verified Code Generation}, 
      author={Lesly Miculicich and Mihir Parmar and Hamid Palangi and Krishnamurthy Dj Dvijotham and Mirko Montanari and Tomas Pfister and Long T. Le},
      year={2025},
      eprint={2510.05156},
      archivePrefix={arXiv},
      primaryClass={cs.SE},
      url={https://arxiv.org/abs/2510.05156}, 
}

@article{Odonoghue2016scs,
author = {O'donoghue, Brendan and Chu, Eric and Parikh, Neal and Boyd, Stephen},
title = {Conic Optimization via Operator Splitting and Homogeneous Self-Dual Embedding},
year = {2016},
issue_date = {June      2016},
publisher = {Plenum Press},
address = {USA},
volume = {169},
number = {3},
issn = {0022-3239},
url = {https://doi.org/10.1007/s10957-016-0892-3},
doi = {10.1007/s10957-016-0892-3},
journal = {Journal of Optimization Theory and Applications},
month = jun,
pages = {1042–1068},
numpages = {27}
}

@inproceedings{deRaedt2007ProbLog,
author = {De Raedt, Luc and Kimmig, Angelika and Toivonen, Hannu},
title = {ProbLog: a probabilistic prolog and its application in link discovery},
year = {2007},
publisher = {Morgan Kaufmann Publishers Inc.},
address = {San Francisco, CA, USA},
booktitle = {Proceedings of the 20th International Joint Conference on Artifical Intelligence},
pages = {2468–2473},
numpages = {6},
location = {Hyderabad, India},
series = {IJCAI'07}
}

@misc{liu2026agentdogdiagnosticguardrailframework,
      title={AgentDoG: A Diagnostic Guardrail Framework for AI Agent Safety and Security}, 
      author={Dongrui Liu and Qihan Ren and Chen Qian and Shuai Shao and Yuejin Xie and Yu Li and Zhonghao Yang and Haoyu Luo and Peng Wang and Qingyu Liu and Binxin Hu and Ling Tang and Jilin Mei and Dadi Guo and Leitao Yuan and Junyao Yang and Guanxu Chen and Qihao Lin and Yi Yu and Bo Zhang and Jiaxuan Guo and Jie Zhang and Wenqi Shao and Huiqi Deng and Zhiheng Xi and Wenjie Wang and Wenxuan Wang and Wen Shen and Zhikai Chen and Haoyu Xie and Jialing Tao and Juntao Dai and Jiaming Ji and Zhongjie Ba and Linfeng Zhang and Yong Liu and Quanshi Zhang and Lei Zhu and Zhihua Wei and Hui Xue and Chaochao Lu and Jing Shao and Xia Hu},
      year={2026},
      eprint={2601.18491},
      archivePrefix={arXiv},
      primaryClass={cs.AI},
      url={https://arxiv.org/abs/2601.18491}, 
}

@inproceedings{wang2019sidechannel,
author = {Wang, Jingbo and Sung, Chungha and Wang, Chao},
title = {Mitigating power side channels during compilation},
year = {2019},
isbn = {9781450355728},
publisher = {Association for Computing Machinery},
address = {New York, NY, USA},
url = {https://doi.org/10.1145/3338906.3338913},
doi = {10.1145/3338906.3338913},
booktitle = {Proceedings of the 2019 27th ACM Joint Meeting on European Software Engineering Conference and Symposium on the Foundations of Software Engineering},
pages = {590–601},
numpages = {12},
location = {Tallinn, Estonia},
series = {ESEC/FSE 2019}
}

@inproceedings{sang2005wmc-complexity,
author = {Sang, Tian and Bearne, Paul and Kautz, Henry},
title = {Performing Bayesian inference by weighted model counting},
year = {2005},
isbn = {157735236x},
publisher = {AAAI Press},
booktitle = {Proceedings of the 20th National Conference on Artificial Intelligence - Volume 1},
pages = {475–481},
numpages = {7},
location = {Pittsburgh, Pennsylvania},
series = {AAAI'05}
}

@inproceedings{
wanders2016judgeD,
title={JudgeD: a probabilistic datalog with dependencies},
author={Brend Wanders and Maurice van Keulen and Jan Flokstra},
booktitle={Workshops at the Thirtieth AAAI Conference on Artificial Intelligence},
year={2016},
}

@article{li2023scallop,
author = {Li, Ziyang and Huang, Jiani and Naik, Mayur},
title = {Scallop: A Language for Neurosymbolic Programming},
year = {2023},
issue_date = {June 2023},
publisher = {Association for Computing Machinery},
address = {New York, NY, USA},
volume = {7},
number = {PLDI},
url = {https://doi.org/10.1145/3591280},
doi = {10.1145/3591280},
journal = {Proceedings of the ACM on Programming Languages},
month = jun,
articleno = {166},
numpages = {25}
}

@inproceedings{manhaeve2018deepproblog,
author = {Manhaeve, Robin and Dumancic, Sebastijan and Kimmig, Angelika and Demeester, Thomas and Raedt, Luc De},
title = {DeepProbLog: neural probabilistic logic programming},
year = {2018},
publisher = {Curran Associates Inc.},
address = {Red Hook, NY, USA},
booktitle = {Proceedings of the 32nd International Conference on Neural Information Processing Systems},
pages = {3753–3763},
numpages = {11},
location = {Montr\'{e}al, Canada},
series = {NIPS'18}
}

@article{fierens2015inference,
  title={Inference and learning in probabilistic logic programs using weighted boolean formulas},
  author={Fierens, Daan and Van den Broeck, Guy and Renkens, Joris and Shterionov, Dimitar and Gutmann, Bernd and Thon, Ingo and Janssens, Gerda and De Raedt, Luc},
  journal={Theory and Practice of Logic Programming},
  volume={15},
  number={3},
  pages={358--401},
  year={2015},
  publisher={Cambridge University Press}
}

@article{holtzen2020dice,
author = {Holtzen, Steven and Van den Broeck, Guy and Millstein, Todd},
title = {Scaling exact inference for discrete probabilistic programs},
year = {2020},
issue_date = {November 2020},
publisher = {Association for Computing Machinery},
address = {New York, NY, USA},
volume = {4},
number = {OOPSLA},
url = {https://doi.org/10.1145/3428208},
doi = {10.1145/3428208},
journal = {Proceedings of the ACM on Programming Languages},
month = nov,
articleno = {140},
numpages = {31}
}

@article{andersen1996linear,
  title={A linear programming framework for logics of uncertainty},
  author={Andersen, Kim Allan and Hooker, John N},
  journal={Decision Support Systems},
  volume={16},
  number={1},
  pages={39--53},
  year={1996},
  publisher={Elsevier}
}

@inproceedings{chen2025struQ,
author = {Chen, Sizhe and Piet, Julien and Sitawarin, Chawin and Wagner, David},
title = {StruQ: defending against prompt injection with structured queries},
year = {2025},
isbn = {978-1-939133-52-6},
publisher = {USENIX Association},
address = {USA},
booktitle = {Proceedings of the 34th USENIX Conference on Security Symposium},
articleno = {123},
numpages = {18},
location = {Seattle, WA, USA},
series = {SEC '25}
}

@inproceedings{chen2025secalign,
author = {Chen, Sizhe and Zharmagambetov, Arman and Mahloujifar, Saeed and Chaudhuri, Kamalika and Wagner, David and Guo, Chuan},
title = {SecAlign: Defending Against Prompt Injection with Preference Optimization},
year = {2025},
isbn = {9798400715259},
publisher = {Association for Computing Machinery},
address = {New York, NY, USA},
url = {https://doi.org/10.1145/3719027.3744836},
doi = {10.1145/3719027.3744836},
booktitle = {Proceedings of the 2025 ACM SIGSAC Conference on Computer and Communications Security},
pages = {2833–2847},
numpages = {15},
location = {Taipei, Taiwan},
series = {CCS '25}
}

@inproceedings{myers1997ifc,
author = {Myers, Andrew C. and Liskov, Barbara},
title = {A decentralized model for information flow control},
year = {1997},
isbn = {0897919165},
publisher = {Association for Computing Machinery},
address = {New York, NY, USA},
url = {https://doi.org/10.1145/268998.266669},
doi = {10.1145/268998.266669},
booktitle = {Proceedings of the Sixteenth ACM Symposium on Operating Systems Principles},
pages = {129–142},
numpages = {14},
location = {Saint Malo, France},
series = {SOSP '97}
}

@article{schneider2000enforceablesecuritypolicies,
author = {Schneider, Fred B.},
title = {Enforceable security policies},
year = {2000},
issue_date = {Feb. 2000},
publisher = {Association for Computing Machinery},
address = {New York, NY, USA},
volume = {3},
number = {1},
issn = {1094-9224},
url = {https://doi.org/10.1145/353323.353382},
doi = {10.1145/353323.353382},
journal = {ACM Transactions on Information and System Security},
month = feb,
pages = {30–50},
numpages = {21}
}

@article{denning1976latticeifc,
author = {Denning, Dorothy E.},
title = {A lattice model of secure information flow},
year = {1976},
issue_date = {May 1976},
publisher = {Association for Computing Machinery},
address = {New York, NY, USA},
volume = {19},
number = {5},
issn = {0001-0782},
url = {https://doi.org/10.1145/360051.360056},
doi = {10.1145/360051.360056},
journal = {Communications of the ACM},
month = may,
pages = {236–243},
numpages = {8}
}

@inproceedings{zheng2024promptdrivensafeguarding,
author = {Zheng, Chujie and Yin, Fan and Zhou, Hao and Meng, Fandong and Zhou, Jie and Chang, Kai-Wei and Huang, Minlie and Peng, Nanyun},
title = {On prompt-driven safeguarding for large language models},
year = {2024},
booktitle = {Proceedings of the 41st International Conference on Machine Learning},
articleno = {2547},
numpages = {21},
location = {Vienna, Austria},
series = {ICML'24}
}

@inproceedings{schick2023toolformer,
author = {Schick, Timo and Dwivedi-Yu, Jane and Dess\'{\i}, Roberto and Raileanu, Roberta and Lomeli, Maria and Hambro, Eric and Zettlemoyer, Luke and Cancedda, Nicola and Scialom, Thomas},
title = {Toolformer: language models can teach themselves to use tools},
year = {2023},
publisher = {Curran Associates Inc.},
address = {Red Hook, NY, USA},
booktitle = {Proceedings of the 37th International Conference on Neural Information Processing Systems},
articleno = {2997},
numpages = {13},
location = {New Orleans, LA, USA},
series = {NIPS '23}
}

@article{sabelfeld2003language,
  title={Language-based information-flow security},
  author={Sabelfeld, Andrei and Myers, Andrew C},
  journal={IEEE Journal on selected areas in communications},
  volume={21},
  number={1},
  pages={5--19},
  year={2003},
  publisher={IEEE}
}

@misc{li2024formalllmintegratingformallanguage,
      title={Formal-LLM: Integrating Formal Language and Natural Language for Controllable LLM-based Agents}, 
      author={Zelong Li and Wenyue Hua and Hao Wang and He Zhu and Yongfeng Zhang},
      year={2024},
      eprint={2402.00798},
      archivePrefix={arXiv},
      primaryClass={cs.LG},
      url={https://arxiv.org/abs/2402.00798}, 
}

@misc{kim2026causalarmorefficientindirectprompt,
      title={CausalArmor: Efficient Indirect Prompt Injection Guardrails via Causal Attribution}, 
      author={Minbeom Kim and Mihir Parmar and Phillip Wallis and Lesly Miculicich and Kyomin Jung and Krishnamurthy Dj Dvijotham and Long T. Le and Tomas Pfister},
      year={2026},
      eprint={2602.07918},
      archivePrefix={arXiv},
      primaryClass={cs.CR},
      url={https://arxiv.org/abs/2602.07918}, 
}

@article{Tsamoura2020groundingbottleneckdatalog, title={Beyond the Grounding Bottleneck: Datalog Techniques for Inference in Probabilistic Logic Programs}, volume={34}, url={https://ojs.aaai.org/index.php/AAAI/article/view/6591}, DOI={10.1609/aaai.v34i06.6591}, abstractNote={&amp;lt;p&amp;gt;State-of-the-art inference approaches in probabilistic logic programming typically start by computing the relevant ground program with respect to the queries of interest, and then use this program for probabilistic inference using knowledge compilation and weighted model counting. We propose an alternative approach that uses efficient Datalog techniques to integrate knowledge compilation with forward reasoning with a non-ground program. This effectively eliminates the grounding bottleneck that so far has prohibited the application of probabilistic logic programming in query answering scenarios over knowledge graphs, while also providing fast approximations on classical benchmarks in the field.&amp;lt;/p&amp;gt;}, number={06}, journal={Proceedings of the AAAI Conference on Artificial Intelligence}, author={Tsamoura, Efthymia and Gutierrez-Basulto, Victor and Kimmig, Angelika}, year={2020}, month={Apr.}, pages={10284–10291} }

@article{barany2017datalogdppl,
author = {B\'{a}R\'{a}ny, Vince and Cate, Balder Ten and Kimelfeld, Benny and Olteanu, Dan and Vagena, Zografoula},
title = {Declarative Probabilistic Programming with Datalog},
year = {2017},
issue_date = {December 2017},
publisher = {Association for Computing Machinery},
address = {New York, NY, USA},
volume = {42},
number = {4},
issn = {0362-5915},
url = {https://doi.org/10.1145/3132700},
doi = {10.1145/3132700},
journal = {ACM Transactions on Database Systems},
month = oct,
articleno = {22},
numpages = {35}
}

@inproceedings{crampton2005referencemonitor,
author = {Crampton, Jason},
title = {A reference monitor for workflow systems with constrained task execution},
year = {2005},
isbn = {1595930450},
publisher = {Association for Computing Machinery},
address = {New York, NY, USA},
url = {https://doi.org/10.1145/1063979.1063986},
doi = {10.1145/1063979.1063986},
booktitle = {Proceedings of the Tenth ACM Symposium on Access Control Models and Technologies},
pages = {38–47},
numpages = {10},
location = {Stockholm, Sweden},
series = {SACMAT '05}
}

@misc{inan2023llamaguardllmbasedinputoutput,
      title={Llama Guard: LLM-based Input-Output Safeguard for Human-AI Conversations}, 
      author={Hakan Inan and Kartikeya Upasani and Jianfeng Chi and Rashi Rungta and Krithika Iyer and Yuning Mao and Michael Tontchev and Qing Hu and Brian Fuller and Davide Testuggine and Madian Khabsa},
      year={2023},
      eprint={2312.06674},
      archivePrefix={arXiv},
      primaryClass={cs.CL},
      url={https://arxiv.org/abs/2312.06674}, 
}

@inproceedings{rebedea-etal-2023-nemo,
    title = "{N}e{M}o Guardrails: A Toolkit for Controllable and Safe {LLM} Applications with Programmable Rails",
    author = "Rebedea, Traian  and
      Dinu, Razvan  and
      Sreedhar, Makesh Narsimhan  and
      Parisien, Christopher  and
      Cohen, Jonathan",
    editor = "Feng, Yansong  and
      Lefever, Els",
    booktitle = "Proceedings of the 2023 Conference on Empirical Methods in Natural Language Processing: System Demonstrations",
    month = dec,
    year = "2023",
    address = "Singapore",
    publisher = "Association for Computational Linguistics",
    url = "https://aclanthology.org/2023.emnlp-demo.40/",
    doi = "10.18653/v1/2023.emnlp-demo.40",
    pages = "431--445",
}

@misc{ghalebikesabi2024operationalizingcontextualintegrityprivacyconscious,
      title={Operationalizing Contextual Integrity in Privacy-Conscious Assistants}, 
      author={Sahra Ghalebikesabi and Eugene Bagdasaryan and Ren Yi and Itay Yona and Ilia Shumailov and Aneesh Pappu and Chongyang Shi and Laura Weidinger and Robert Stanforth and Leonard Berrada and Pushmeet Kohli and Po-Sen Huang and Borja Balle},
      year={2024},
      eprint={2408.02373},
      archivePrefix={arXiv},
      primaryClass={cs.AI},
      url={https://arxiv.org/abs/2408.02373}, 
}

@inproceedings{
li2026drift,
title={{DRIFT}: Dynamic Rule-Based Defense with Injection Isolation for Securing {LLM} Agents},
author={Hao Li and Xiaogeng Liu and CHIU Hung Chun and Dianqi Li and Ning Zhang and Chaowei Xiao},
booktitle={The Thirty-ninth Annual Conference on Neural Information Processing Systems},
year={2026},
url={https://openreview.net/forum?id=oY1Xnt83oJ}
}

@inproceedings{tsai2025conesca,
author = {Tsai, Lillian and Bagdasarian, Eugene},
title = {Contextual Agent Security: A Policy for Every Purpose},
year = {2025},
isbn = {9798400714757},
publisher = {Association for Computing Machinery},
address = {New York, NY, USA},
url = {https://doi.org/10.1145/3713082.3730378},
doi = {10.1145/3713082.3730378},
booktitle = {Proceedings of the 2025 Workshop on Hot Topics in Operating Systems},
pages = {8–17},
numpages = {10},
location = {Banff, AB, Canada},
series = {HotOS '25}
}

@misc{shi2026progentsecuringaiagents,
      title={Progent: Securing AI Agents with Privilege Control}, 
      author={Tianneng Shi and Jingxuan He and Zhun Wang and Hongwei Li and Linyu Wu and Wenbo Guo and Dawn Song},
      year={2026},
      eprint={2504.11703},
      archivePrefix={arXiv},
      primaryClass={cs.CR},
      url={https://arxiv.org/abs/2504.11703}, 
}

@misc{zhu2025miniscopeprivilegeframeworkauthorizing,
      title={MiniScope: A Least Privilege Framework for Authorizing Tool Calling Agents}, 
      author={Jinhao Zhu and Kevin Tseng and Gil Vernik and Xiao Huang and Shishir G. Patil and Vivian Fang and Raluca Ada Popa},
      year={2025},
      eprint={2512.11147},
      archivePrefix={arXiv},
      primaryClass={cs.CR},
      url={https://arxiv.org/abs/2512.11147}, 
}

@inproceedings{
zhang2025agent,
title={Agent Security Bench ({ASB}): Formalizing and Benchmarking Attacks and Defenses in {LLM}-based Agents},
author={Hanrong Zhang and Jingyuan Huang and Kai Mei and Yifei Yao and Zhenting Wang and Chenlu Zhan and Hongwei Wang and Yongfeng Zhang},
booktitle={The Thirteenth International Conference on Learning Representations},
year={2025},
url={https://openreview.net/forum?id=V4y0CpX4hK}
}

@inproceedings{zhan2024injecagent,
  title={Injecagent: Benchmarking indirect prompt injections in tool-integrated large language model agents},
  author={Zhan, Qiusi and Liang, Zhixiang and Ying, Zifan and Kang, Daniel},
  booktitle={Findings of the Association for Computational Linguistics: ACL 2024},
  pages={10471--10506},
  year={2024}
}

@inproceedings{
debenedetti2024agentdojo,
title={AgentDojo: A Dynamic Environment to Evaluate Prompt Injection Attacks and Defenses for {LLM} Agents},
author={Edoardo Debenedetti and Jie Zhang and Mislav Balunovic and Luca Beurer-Kellner and Marc Fischer and Florian Tram{\`e}r},
booktitle={The Thirty-eight Conference on Neural Information Processing Systems Datasets and Benchmarks Track},
year={2024},
url={https://openreview.net/forum?id=m1YYAQjO3w}
}

@online{microsoft2026rce-vulnerabilities,
  author    = {Microsoft Defender Security Research Team and Uri Oren and Amit Eliahu and Dor Edry},
  title     = {When prompts become shells: RCE vulnerabilities in AI agent frameworks},
  url       = {https://www.microsoft.com/en-us/security/blog/2026/05/07/prompts-become-shells-rce-vulnerabilities-ai-agent-frameworks/},
  journal   = {Microsoft Security Blog},
  year      = {2026},
  month     = may,
  urldate   = {2026-05-07}
}

@inproceedings{park2023generativagents,
author = {Park, Joon Sung and O'Brien, Joseph and Cai, Carrie Jun and Morris, Meredith Ringel and Liang, Percy and Bernstein, Michael S.},
title = {Generative Agents: Interactive Simulacra of Human Behavior},
year = {2023},
isbn = {9798400701320},
publisher = {Association for Computing Machinery},
address = {New York, NY, USA},
url = {https://doi.org/10.1145/3586183.3606763},
doi = {10.1145/3586183.3606763},
booktitle = {Proceedings of the 36th Annual ACM Symposium on User Interface Software and Technology},
articleno = {2},
numpages = {22},
location = {San Francisco, CA, USA},
series = {UIST '23}
}

@online{bourtoule2026isolationoldvulnerabilities,
  author    = {Lucas Bourtoule},
  title     = {Lack of isolation in agentic browsers resurfaces old vulnerabilities},
  url       = {https://blog.trailofbits.com/2026/01/13/lack-of-isolation-in-agentic-browsers-resurfaces-old-vulnerabilities/},
  journal   = {The Trail of Bits Blog},
  year      = {2026},
  month     = may,
  urldate   = {2026-01-13}
}

@InProceedings{jordan2016souffle,
author="Jordan, Herbert
and Scholz, Bernhard
and Suboti{\'{c}}, Pavle",
editor="Chaudhuri, Swarat
and Farzan, Azadeh",
title="Souffl{\'e}: On Synthesis of Program Analyzers",
booktitle="Computer Aided Verification",
year="2016",
publisher="Springer International Publishing",
address="Cham",
pages="422--430",
isbn="978-3-319-41540-6"
}

@inproceedings{
merrill2026terminalbench,
title={Terminal-Bench: Benchmarking Agents on Hard, Realistic Tasks in Command Line Interfaces},
author={Mike A Merrill and Alexander Glenn Shaw and Nicholas Carlini and Boxuan Li and Harsh Raj and Ivan Bercovich and Lin Shi and Jeong Yeon Shin and Thomas Walshe and E. Kelly Buchanan and Junhong Shen and Guanghao Ye and Haowei Lin and Jason Poulos and Maoyu Wang and Marianna Nezhurina and Di Lu and Orfeas Menis Mastromichalakis and Zhiwei Xu and Zizhao Chen and Yue Liu and Robert Zhang and Leon Liangyu Chen and Anurag Kashyap and Jan-Lucas Uslu and Jeffrey Li and Jianbo Wu and Minghao Yan and Song Bian and Vedang Sharma and Ke Sun and Steven Dillmann and Akshay Anand and Andrew Lanpouthakoun and Bardia Koopah and Changran Hu and Etash Kumar Guha and Gabriel H. S. Dreiman and Jiacheng Zhu and Karl Krauth and Li Zhong and Niklas Muennighoff and Robert Kwesi Amanfu and Shangyin Tan and Shreyas Pimpalgaonkar and Tushar Aggarwal and Xiangning Lin and Xin Lan and Xuandong Zhao and Yiqing Liang and Yuanli Wang and Zilong Wang and Changzhi Zhou and David Heineman and Hange Liu and Harsh Trivedi and John Yang and Junhong Lin and Manish Shetty and Michael Yang and Nabil Omi and Negin Raoof and Shanda Li and Terry Yue Zhuo and Wuwei Lin and Yiwei Dai and Yuxin Wang and Wenhao Chai and Shang Zhou and Dariush Wahdany and Ziyu She and Jiaming Hu and Zhikang Dong and Yuxuan Zhu and Sasha Cui and Ahson Saiyed and Arinbj{\"o}rn Kolbeinsson and Christopher Michael Rytting and Ryan Marten and Yixin Wang and Jenia Jitsev and Alex Dimakis and Andy Konwinski and Ludwig Schmidt},
booktitle={The Fourteenth International Conference on Learning Representations},
year={2026},
url={https://openreview.net/forum?id=a7Qa4CcHak}
}

@inproceedings{
    jimenez2024swebench,
    title={{SWE}-bench: Can Language Models Resolve Real-world Github Issues?},
    author={Carlos E Jimenez and John Yang and Alexander Wettig and Shunyu Yao and Kexin Pei and Ofir Press and Karthik R Narasimhan},
    booktitle={The Twelfth International Conference on Learning Representations},
    year={2024},
    url={https://openreview.net/forum?id=VTF8yNQM66}
}

@InCollection{
scip2008,
title = {Constraint Integer Programming: a New Approach to
  Integrate CP and MIP},
author = {Tobias Achterberg and Timo Berthold and Thorsten Koch and
  Kati Wolter},
booktitle = {Integration of AI and OR Techniques in Constraint
  Programming for Combinatorial Optimization Problems},
series = {Lecture Notes in Computer Science},
volume = {5015},
pages = {6--20},
year = {2008},
publisher = {Springer},
doi = {10.1007/978-3-540-68155-7_4},
}

@InProceedings{demouraz32008,
author="de Moura, Leonardo
and Bj{\o}rner, Nikolaj",
editor="Ramakrishnan, C. R.
and Rehof, Jakob",
title="Z3: An Efficient SMT Solver",
booktitle="Tools and Algorithms for the Construction and Analysis of Systems",
year="2008",
publisher="Springer Berlin Heidelberg",
address="Berlin, Heidelberg",
pages="337--340",
isbn="978-3-540-78800-3"
}

@inproceedings{10.1145/586110.586145,
author = {Wagner, David and Soto, Paolo},
title = {Mimicry attacks on host-based intrusion detection systems},
year = {2002},
isbn = {1581136129},
publisher = {Association for Computing Machinery},
address = {New York, NY, USA},
url = {https://doi.org/10.1145/586110.586145},
doi = {10.1145/586110.586145},
booktitle = {Proceedings of the 9th ACM Conference on Computer and Communications Security},
pages = {255–264},
numpages = {10},
location = {Washington, DC, USA},
series = {CCS '02}
}

@techreport{Anderson:1972,
  author = {Anderson, James P.},
  institution = {U.S. Air Force Electronic Systems Division},
  month = {10},
  number = {ESD-TR-73-51},
  title = {{C}omputer {S}ecurity {T}echnology {P}lanning {S}tudy},
  volume = 2,
  year = 1972
}

@article{mccormick1976computability,
author = {Mccormick, Garth P.},
title = {Computability of global solutions to factorable nonconvex programs: Part I -- Convex underestimating problems},
year = {1976},
issue_date = {December  1976},
publisher = {Springer-Verlag},
address = {Berlin, Heidelberg},
volume = {10},
number = {1},
issn = {0025-5610},
url = {https://doi.org/10.1007/BF01580665},
doi = {10.1007/BF01580665},
journal = {Mathematical Programming},
month = dec,
pages = {147–175},
numpages = {29}
}

@book{Denning1982-qb,
  title     = "Cryptography and Data Security",
  author    = "Denning, Dorothy Elizabeth Robling",
  publisher = "Addison-Wesley",
  month     =  dec,
  year      =  1982,
  address   = "Boston, MA",
  language  = "en"
}

@INPROCEEDINGS{gray1991informationflowsecurity,
  author={Gray, J.W.},
  booktitle={Proceedings. 1991 IEEE Computer Society Symposium on Research in Security and Privacy}, 
  title={Toward a mathematical foundation for information flow security}, 
  year={1991},
  volume={},
  number={},
  pages={21-34},
  doi={10.1109/RISP.1991.130769}}

\end{document}